\documentclass[12pt]{article}

\usepackage{float}
\usepackage{amsmath}
\usepackage{amsfonts}
\usepackage{amssymb}
\usepackage{graphicx}
\usepackage{gauss}
\usepackage{amsthm}
\usepackage{lscape}
\usepackage{enumerate}
\usepackage[noadjust]{cite}
\usepackage{tabularx}
\usepackage{hyperref}% add hypertext capabilities
\usepackage{cleveref}
\usepackage{caption,subcaption}
\usepackage{color}
\usepackage{geometry}
\usepackage{longtable}
\usepackage{makecell}
\usepackage{comment}
\def\keywords#1{\par\addvspace\medskipamount{\rightskip=0pt plus1cm
\def\and{\ifhmode\unskip\nobreak\fi\ $\cdot$
}\noindent\keywordname\enspace\ignorespaces#1\par}}

\newtheorem{theorem}{Theorem}[section]
\newtheorem{lemma}[theorem]{Lemma}
\newtheorem{definition}[theorem]{Definition}
\newtheorem{corollary}[theorem]{Corollary}

\newtheorem{remark}[theorem]{Remark}
\newtheorem{remarks}[theorem]{Remarks}

\newtheorem{example}[theorem]{Example}
\newtheorem{excont}[theorem]{Example}

\newcommand{\ignore}[1]{}

\newcommand{\END}{\hfill\mbox{\raggedright$\Diamond$}}
\newcommand{\etal}{\textrm{et al.\,}}

\newcommand*{\doi}[1]{\href{http://dx.doi.org/#1}{doi: #1}}

\newcommand{\Matrixc}[1]{\ensuremath{\left[\begin{array}{cccccccccccc} #1 \end{array}\right]}}

\numberwithin{equation}{section}

\renewcommand{\epsilon}{\varepsilon}

\newcommand{\R}{{\mathbf R}}

\renewcommand{\AA}{{\mathcal A}}

\newcommand{\BB}{{\mathcal B}}
\newcommand{\CC}{{\mathcal C}}

\newcommand{\GG}{{\mathcal G}}

\newcommand{\II}{{\mathcal I}}
\newcommand{\KK}{{\mathcal K}}

\newcommand{\NN}{{\mathcal N}}

\newcommand{\PP}{{\mathcal P}}

\newcommand{\LL}{{\mathcal L}}
\newcommand{\RR}{{\mathcal R}}
\newcommand{\sS}{{\mathcal S}}

\newcommand{\VV}{{\mathcal V}} %

\newcommand{\tLL}{{\til{\LL}}}

\newcommand{\til}[1]{\widetilde{#1}}

\newcommand{\pathto}{\rightsquigarrow}

\newcommand{\selfarrow}{\mbox{\rotatebox[origin=c]{90}{$\circlearrowright$}}}

\newcommand{\flatrightarrow}{\raisebox{-.2ex}{\mbox{\;\rotatebox{90}{\scalebox{1}[1.2]{$\bot$}}}\;\;}}

\newcommand{\biarrow}{\;\smash{{}^{\displaystyle \rightarrow}_{\displaystyle\leftarrow}}\;}

\usepackage{chemfig}

\makeatletter
\renewenvironment{description}%
{\list{}{\leftmargin=0pt
\labelwidth\z@ \itemindent-\leftmargin
}}%
{\endlist}
\makeatother

\makeatletter
\def\and{%
  \end{tabular}%
  \hskip 1em \@plus.17fil\relax
  \begin{tabular}[t]{c}}
\makeatother

\title{Homeostasis in Gene Regulatory Networks}

\author{Fernando Antoneli\renewcommand{\thefootnote}{\arabic{footnote}}\footnotemark[1] \textsuperscript{,}\renewcommand{\thefootnote}{\fnsymbol{footnote}}\footnotemark[1]  \and 
Martin Golubitsky\renewcommand{\thefootnote}{\arabic{footnote}}\footnotemark[2] \and 
Jiaxin Jin\renewcommand{\thefootnote}{\arabic{footnote}}\footnotemark[2] \textsuperscript{,}\renewcommand{\thefootnote}{\fnsymbol{footnote}}\footnotemark[1] \and 
Ian Stewart\renewcommand{\thefootnote}{\arabic{footnote}}\footnotemark[3]}

\date{\today}

\begin{document}

\maketitle

\renewcommand{\thefootnote}{\arabic{footnote}}
\footnotetext[1]{Centro de Bioinform\'atica M\'edica, Universidade Federal de S\~ao Paulo, S\~ao Paulo, SP, Brazil}
\footnotetext[3]{Department of Mathematics, The Ohio State University, Columbus, OH, USA}
\footnotetext[4]{Mathematics Institute, University of Warwick, Coventry, UK}

\renewcommand{\thefootnote}{\fnsymbol{footnote}}
\footnotetext[1]{Correspondence:
\href{mailto:jin.1307@osu.edu}{jin.1307@osu.edu},
\href{mailto:fernando.antoneli@unifesp.br}{fernando.antoneli@unifesp.br}
}

\begin{abstract}
Gene regulatory networks lie at the heart of many important intracellular signal transduction processes. 
A Gene Regulatory Network (GRN) is abstractly defined as a directed graph, where the nodes represent genes and the edges represent the causal regulatory interactions between genes.
It can be used to construct mathematical models describing the time-varying concentrations of the several molecular species attached to each gene (node) in the network.
In the deterministic setting this is typically implemented by a system of ordinary differential equations. 

A biological system exhibits homeostasis when there is a target quantity, called the input-output function, whose values stay within a narrow range, under relatively wide variation of an external parameter.
A strong form of homeostasis, called infinitesimal homeostasis, occurs when the input-output function has a critical point.

In this paper, we use the framework of infinitesimal homeostasis to study general design principles for the occurrence of homeostasis in gene regulatory networks.
We assume that the dynamics of the genes explicitly includes both transcription and translation, keeping track of both mRNA and protein concentrations.
Given a GRN we construct an associated Protein-mRNA Network (PRN), where each individual (mRNA and protein) concentration corresponds to a node and the edges are defined in such a way that the PRN becomes a bipartite directed graph.
By simultaneously working with the GRN and the PRN we are able to apply our previous results about the classification of homeostasis types (i.e., topologically defined homeostasis generating mechanism) and their corresponding homeostasis patterns.
Given an arbitrarily large and complex GRN $\mathcal{G}$ and its associated PRN $\mathcal{R}$, we obtain a correspondence between all the homeostasis types (and homeostasis patterns) of $\mathcal{G}$ and a subset the homeostasis types (and homeostasis patterns) of $\mathcal{R}$.
Moreover, we completely characterize the 
homeostasis types of the PRN that do not have GRN counterparts.

\smallskip

\end{abstract}

%%%%%%%%%%%%%%%%%%%%%%%%%%%%%%%%%%%%%%
\newpage

\tableofcontents
%%%%%%%%%%%%%%%%%%%%%%%%%%%%%%%%%%%%%%

\section{Introduction}
\label{S:intro}

{\em Gene expression} is the process by which the information encoded in a gene is turned into a biological function, which ultimately manifests itself as a phenotype effect.
This is accomplished by a complex series of enzymatic chemical reactions within the cell leading to the synthesis of specific macro-molecules called the {\em gene product}.
The process of gene expression is used by all known life -- eukaryotes (including multicellular organisms), prokaryotes (bacteria and archaea), and even viruses -- to generate the molecular machinery for life.

There are, basically, two types of gene products: (i) for {\em protein-coding genes} the gene product is a {\em protein}; (ii) {\em non-coding genes}, such as transfer RNA (tRNA) and small nuclear RNA (snRNA), the gene product is a functional {\em non-coding} RNA (ncRNA).

{\em Regulation of gene expression}, or simply {\em gene regulation}, is the range of mechanisms that are used by cells to increase or decrease the amount of specific gene products. 
Sophisticated schemes of gene expression are widely observed in biology, going from triggering developmental pathways, to responding to environmental stimuli.

A {\em gene} (or {\em genetic}) {\em regulatory network} (GRN) is a collection of molecular regulators that interact with each other and with other substances in the cell to govern the gene expression levels of mRNA and proteins. 
The {\em molecular regulators} can be DNA, RNA, protein or any combination of two, or more of these three that form a complex, such as a specific sequence of DNA and a transcription factor to activate that sequence. 
The interaction can be direct or indirect (through transcribed RNA or translated protein).
When a protein acts as a regulator it is called a {\em transcription factor}, which is one of the main players in regulatory networks.  
By binding to the promoter region of a coding gene they turn them on, initiating the production of another protein, and so on. 
Transcription factors can be {\em excitatory} ({\em activators}) or {\em inhibitory} ({\em repressors}).

Another fundamental concept in biology is that of {\em homeostasis}, which derives from the Greek language, and means `to maintain a similar stasis' \cite{C39}.  
A prototypical biological example of homeostasis occurs in warm blooded mammals where the mammal's internal body temperature remains approximately constant on variation of the external temperature. 

The notion of homeostasis is often associated with regulating global physiological parameters like temperature, hormone levels, or concentrations of molecules in the bloodstream in complex multicellular organisms. 
However, it also can be applied to unicellular organisms, where the issue is how some internal cell state of interest (such as the concentration of some gene product) responds to changes in the intra-cellular and/or extra-cellular environment \cite{MTELT09,TM16,SMXZT17}.
For instance, Antoneli~\etal~\cite{AGS18} study the occurrence of homeostasis in a feedforward loop motif from the GRN of \emph{S.~cerevisiae} (see Figure \ref{F:reg_net_1a}).

In a series of papers about homeostasis in systems of ODEs we have developed a comprehensive theory for its analysis and classification:
(1) homeostasis can be formulated in terms of infinitesimal homeostasis and singularity theory \cite{GS17,GS18,GS23},
(2) infinitesimal homeostasis in biochemical networks where nodes represent one-dimensional concentrations of substrates can be studied in an abstract framework of `input-output networks' \cite{RBGSN17, GW20, WHAG21},
(3) infinitesimal homeostasis can be topologically characterized in terms of `homeostasis types' on a general class of input-output networks \cite{WHAG21,MF21,MF22}, and
(4) homeostasis types themselves can be classified in terms of `homeostasis patterns' \cite{DABGNRS23}. 

In this paper we build on these results to deal with infinitesimal homeostasis in gene regulatory networks (GRN).
As we explain in Section \ref{SS:GRN} a GRN is not exactly a biochemical network as in \cite{RBGSN17}.
There is a `mismatch' between the number of nodes in the network and the number of state variables of the underlying system of ODEs.
In order to resolve this `mismatch',
we generalize the approach of \cite{AGS18}, used to analyze feedforward loops, to arbitrary GRNs.
The idea there was to replace each `gene' node of a GRN by a pair of `protein' and `mRNA' nodes to obtain a {\em protein-mRNA network} (PRN).
Now, PRNs have a mathematical structure similar to that of biochemical networks and hence the theory of \cite{WHAG21,DABGNRS23} can be readily applied.
More importantly, since the classification results have a purely combinatorial side, we can consider the GRN and its associated PRN simultaneously, and work out a correspondence between the classifications of homeostasis types and homeostasis patterns in both of them.
Even though infinitesimal homeostasis only makes sense on the PRN (dynamical) level, its purely combinatorial aspects can be transferred to the GRN.

The main result of this paper is a complete characterization of homeostasis types and homeostasis patterns on the PRN that have correspondent on the GRN.
A byproduct of this characterization is the discovery of homeostasis types and homeostasis patterns on the PRN that do not have counterparts on the GRN.
The novelty of our approach is the simultaneous use of two networks, the GRN and the PRN, in the analysis of gene expression homeostasis, and the lack of assumptions about the functional form of the differential equations.

\subsection{Mathematical Modeling of Gene Regulatory Networks}

The development of advanced experimental techniques in molecular biology is producing increasingly large amounts of experimental data on gene regulation. 
This, in turn, demands the development of mathematical modelling methods for the study and analysis of gene regulation.
Mathematical models of GRNs describe both gene expression and regulation, and in some cases generate predictions that support experimental observations.
Formally, a GRN is represented by a directed graph.
Nodes represent the variables associated to genes (e.g., mRNA and/or protein concentration) and directed links represent couplings between genes (e.g., effect of one gene product on other genes).
In any case, GRN models can be roughly divided into three classes (see \cite{KS08}):
\begin{enumerate}[(1)]
\item {\it Logical models.} This class of models aims to describe regulatory networks qualitatively. They allow users to obtain a basic understanding of the different functionalities of a given network under different conditions. Their qualitative nature makes them flexible and easy to fit to biological phenomena, although they can only answer qualitative questions. 
Among them, the most common approaches are those based on Boolean networks \cite{K69,T73,SS96}. 
See also Barbuti \etal \cite{BGMN20} for a comprehensive review.
\item {\it Continuous models.} This class of models allows us to understand and manipulate behaviours that depend on finer timing and exact molecular concentrations. For example, to simulate the effects of dietary restriction on yeast cells under different nutrient concentrations, users must resort to the finer resolution of continuous models.
These models are best formulated in terms of nonlinear dynamical systems given by coupled systems of ordinary differential equations \cite{G63,G68a,G68b,TO78,S87}.
See also Polynakis \etal \cite{PHB09}.
\item {\it Stochastic models.} This class of models was introduced following the observation that the functionality of a regulatory network is often affected by noise \cite{ELSS02}. 
As the majority of these models account for interactions between individual molecules, they are called single molecule level models \cite{K91,PY95,MA97,KE01,KEBC05}.
See also Bocci \etal \cite{BJNJO23}.
\end{enumerate}

This paper focuses on the mathematical modelling of GRNs using coupled systems of ordinary differential equations.
In the coupled ODE setting there is an  important issue concerning the number of variables / equations associated to each node.
We will return to this issue in Section
\ref{SS:GRN}.

\subsection{Gene Expression Homeostasis}
\label{SS:GEH}

Homeostasis is often modeled using differential equations and, in this context, can be interpreted in two mathematically distinct ways.  
One boils down to the existence of a `globally stable equilibrium'. 
Here changes in the environment are considered to be perturbations of {\em initial conditions} \cite{K04,K07,TD90}. 
A stronger interpretation works with a {\em parametrized} family of differential equations possessing a {\em stable equilibrium}.
Now homeostasis means that a function of this equilibrium, the {\em input-output function}, changes by a relatively small amount when the {\em parameter} varies by a much larger amount \cite{NBR14,TM16,GS17}.
In this paper we adopt the second, stronger, interpretation and focus on the mathematical aspects of this concept. 

There is a large body of work about homeostasis from the standpoint of control theory. 
In this context, homeostasis is related to the more stringent notion of {\em robust perfect adaptation}.
Now, the input-output function is exactly constant over the parameter range (see \cite{F16,K21,ST23} for more details).  
We will not consider this stricter form of homeostasis here, but it is worth pointing out that our results do apply to this context, since robust perfect adaptation is a particular case of the notion of `infinitesimal homeostasis' (see \cite{MF21,MF22}).

Given a family of differential equations depending on a parameter $\II$ and possessing a stable equilibrium $X(\II)$,
we say that {\em homeostasis} occurs in this system if the input-output function $z(\II)=\Phi(X(\II))$, where $\Phi$ a smooth function, is approximately constant upon variation of $\II$.  
Golubitsky and Stewart~\cite{GS17} observe that homeostasis on some neighborhood of a specific value $\II_0$ follows from the occurrence of {\em infinitesimal homeostasis} at $\II_0$: $\frac{dz}{d\II}(\II_0) = 0$.  
This observation is essentially the well-known fact that the value of a function changes most slowly near a stationary (or critical) point.
Despite the name, `infinitesimal homeostasis' often implies that the system is homeostatic over a relatively large interval of the parameter \cite[Section 5.4]{GS23}.  
The key quantity is the value of the {\em second derivative} of $z$ at $\II_0$. 
As the second derivative becomes smaller, the interval of homeostasis becomes larger.

In Antoneli \etal \cite{AGS18} the authors apply this formalism to find infinitesimal homeostasis in a small $3$-node GRN called `feedforward loop motif' (see Figure \ref{F:reg_net_1a}).
Assuming that the regulation of the three genes is inhibitory (repression) and that only gene SPF1 is regulated by upstream transcription factors (the input parameter), they show that the protein concentration of gene GAP1 (the output node) robustly exhibits infinitesimal homeostasis with respect to variation on the regulation level of SPF1 over a wide range.
Moreover, SPF1 and GZF3 (i.e., their protein concentrations) are not homeostatic for any value of the input parameter.
Here, `robustly' means that the occurrence of infinitesimal homeostasis (or not) on the protein concentrations of the three genes described above is persistent under variation of kinetic parameters of the defining differential equations (the rates of synthesis and degradation of the mRNA and protein concentrations).
In order to obtain compatibility of the infinitesimal homeostasis formalism with the GRN structure they use the protein-mRNA network (PRN) representation.
Moreover, in this particular setting, they were able to explicitly compute the homeostasis point by assuming a special functional form for the equations.

In this paper, we consider arbitrary large and complex GRN, with the most general functional form for the dynamics.
However, in this generality, it is no longer possible to explicitly compute homeostasis points.

\begin{figure}[!htb]
\centerline{\includegraphics[width=0.4\textwidth]{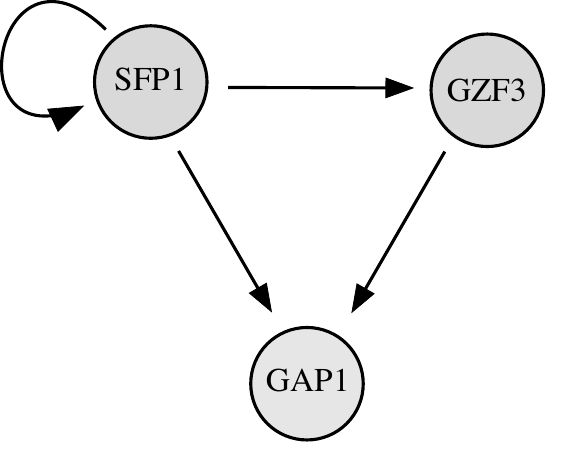}}
\caption{Feedforward loop motif from the GRN of \emph{S.~cerevisiae} (see Antoneli \etal \cite{AGS18}).  Note that autoregulation implies that the protein concentration associated with the gene SFP1 affects directly the mRNA concentration associated with that gene.} 
\label{F:reg_net_1a}
\end{figure}

\paragraph{Structure of the Paper.}
In Section~\ref{SS:GRN} we introduce the notion of GRN and its associated PRN.
We recall the definition of input-output network, infinitesimal homeostasis, homeostasis types, homeostasis subnetworks and homeostasis patterns and show that this theory can be directly applied to PRNs.
Next, we explain how the purely combinatorial part of the classification theory can be applied directly to the GRN and how this relates to the classification for the PRN.
We illustrate the general theory with two paradigmatic examples: (a) feedforward loop and (b) 
feedback inhibition.
The proofs of all the results in full generality are provided in the Appendix. 
In Appendix \ref{S:Graphic_theory}, we recall the basic terminology of input-output networks and state the results regarding the combinatorial characterization of homeostasis types and homeostasis patterns.
In Appendix \ref{S:inf homeo in RPN}, we prove the results about infinitesimal homeostasis classification in the GRN and its associated PRN are related to each other.
In Appendix \ref{S:pattern of homeo}, we show how the structure of homeostasis patterns in the GRN and its associated PRN are related to each other.

\section{Gene Regulatory Networks}
\label{SS:GRN}

At the abstract level a GRN is a directed graph whose nodes are the genes and a directed link from a source gene to a target gene indicates that the gene product of the source gene acts as a molecular regulator of the target gene.
{\em Autoregulation} occurs when a link connects a gene to itself, that is, the gene product is a molecular regulator of the gene itself.

At the dynamical level, a gene (a node in the GRN) represents the collection of processes that ultimately lead to the making of the gene product.
A protein-coding gene should have at least two processes: (i) {\em transcription}, that is, the synthesis of a mRNA copy from a DNA strand and (ii) {\em translation}, that is, the synthesis of protein from a mRNA. 
The `output' of the corresponding node in the GRN is the protein concentration at a given time.

Assume the simplest scenario, namely, a GRN containing only protein-coding genes.
Then, each gene (node) represents two concentrations (the concentration of mRNA and the concentration of protein). 
That is, there is a mismatch between the number of nodes $N$ and the number of state variables $2N$.
There are two ways to deal with this mismatch.
\begin{enumerate}[(1)]
\item {\it Protein-protein formalism.} This approach is based on the fact that, in some cases, the changes to mRNA concentrations occur much faster than the changes to the concentrations of the associated proteins \cite{PHB09}. 
More specifically, the mRNA concentration quickly reaches a steady-state value before any protein is translated from it.
Formally, this technique is called {\em quasi steady-state approximation} (QSSA) \cite{SGT17}.
Then, we can solve the steady-state mRNA equations and plug the result in the protein equations.
This procedure effectively reduces the number of state variables by half, thus matching of the number of nodes in the GRN (see \cite{KS08,RS17a,RS17b,SMXZT17}).

\item {\it Protein-mRNA formalism.}
In this approach we keep the mRNA and protein concentrations for each gene and double the number of nodes of the network, leading to the notion of {\em protein-mRNA network} (PRN) \cite{PHB09,MSTZ16}.
Now the network has two `types' of nodes (mRNA and protein) and two `types' of arrows (mRNA $\to$ protein and protein $\to$ mRNA).
As we will see below there is a correspondence between the two networks that allows us to transfer some properties back and forth. 
For instance, this is the approach adopted in Antoneli~\etal~\cite{AGS18} for the particular example of feedforward loop motif (see also \cite{KEBC05,M05,M08}).
\end{enumerate}

There are mathematical and biological reasons to prefer the second possibility.
From the mathematical point of view it is more convenient to work with the PRN \cite{MMRS23}. 
More specifically, it allows us to use the general theory of network dynamics \cite{GS23} to associate a natural class of ODEs to a PRN, which contains virtually all models discussed in the literature. 
Moreover, it makes it possible to apply the techniques developed in Wang \etal~\cite{WHAG21} and Duncan \etal~\cite{DABGNRS23} to classify the homeostasis types in PRN and GRN.

Biologically, the use of protein-protein networks is more appropriate to model prokaryotic gene regulation, due to the fact that: (i) transcription and translation occur, almost simultaneously, within the cytoplasm of a cell due to the lack of a defined nucleus, (ii) the coding regions typically take up $\sim 90\%$ of the genome, whereas the remaining $\sim 10\%$ does not encode proteins, but most of it still has some biological function (e.g., genes for transfer RNA and ribosomal RNA).
In this case, it is reasonable to assume that gene expression is regulated primarily at the transcriptional level.
On the other hand, gene expression in eukaryotes is a much more complicated process, with several intermediate steps: transcription occurs in the nucleus, where mRNA is processed, modified and transported, and translation can occur in a variety of regions of the cell.
In particular, several non-coding genes that are transcribed into functional non-coding RNA molecules, such as, microRNAs (miRNAs), short interfering RNAs (siRNAs) and long non-coding RNAs (lncRNAs), play an important role in eukaryotic gene regulation.
As we will explain later, it is very easy to incorporate these regulatory elements into the framework of PRNs (see Remark \ref{RMK:non_coding}).

\subsection{From GRN to PRN} 
\label{SS:PRN}

From now on we make the simplifying assumption that the GRN consists of protein-coding genes with transcription and translation.

Since we are interested in gene expression homeostasis, we consider {\em input-output} GRNs. 
They are supplied with an external parameter $\II$ -- e.g., environmental disturbance or transcription activity (a function of the concentration of transcription factors) -- that affects the mRNA transcription of one gene, called the {\em input gene} $\iota$ of the GRN.
The protein concentration of a second gene, called the {\em output gene} $o$ of the GRN, is the concentration where we expect to exhibit homeostasis.
These two distinguished nodes are fixed throughout the analysis.
We assume that only the input node is affected by the external parameter $\II$.

Given a GRN $\GG$ we construct an associated PRN  $\RR$ as follows.  
Every node $\rho$ in the GRN $\GG$ corresponds to two nodes in the associated PRN: $\rho^R$ (the mRNA concentration of gene $\rho$) and $\rho^P$ (the protein concentration of gene $\rho$).
Since there is no intermediary process, the protein concentration $\rho^P$ is affected only by the mRNA concentration $\rho^R$. 
Hence there is a PRN arrow from $\rho^R\to\rho^P$ and no other PRN arrow has head node $\rho^P$. 
Next, each GRN arrow $\sigma\to\rho$ leads to a PRN arrow from the protein concentration $\sigma^P$ to the mRNA concentration $\rho^R$ (that is, $\sigma$ is a transcription factor of $\rho$).
Note that each arrow in the GRN leads to a single arrow in the PRN.  
In particular, autoregulation in $\rho$ leads to an arrow of $\rho^P\to\rho^R$.  
Finally, if the GRN is an input-output network with input gene $\iota$ and output gene $o$, then the associated PRN is an input-output network with input node $\iota^R$ and output node $o^P$.
It follows that a PRN is always a {\em bipartite digraph (directed graph)} \cite{PKPBMB18}, where the two distinguished subsets of nodes are the $\rho^R$ nodes and the $\rho^P$ nodes.

For example, the abstract GRN corresponding to the feedforward loop motif shown in Figure~\ref{F:reg_net_1a} is the $3$-node input-output network shown in Figure~\ref{F:reg_net_1b}(a). 
Its associated PRN is the $6$-node input-output network shown in Figure~\ref{F:reg_net_1b}(b).

\begin{figure}[!htb]
\begin{subfigure}[b]{0.45\textwidth}
\centering
\includegraphics[width=\textwidth]{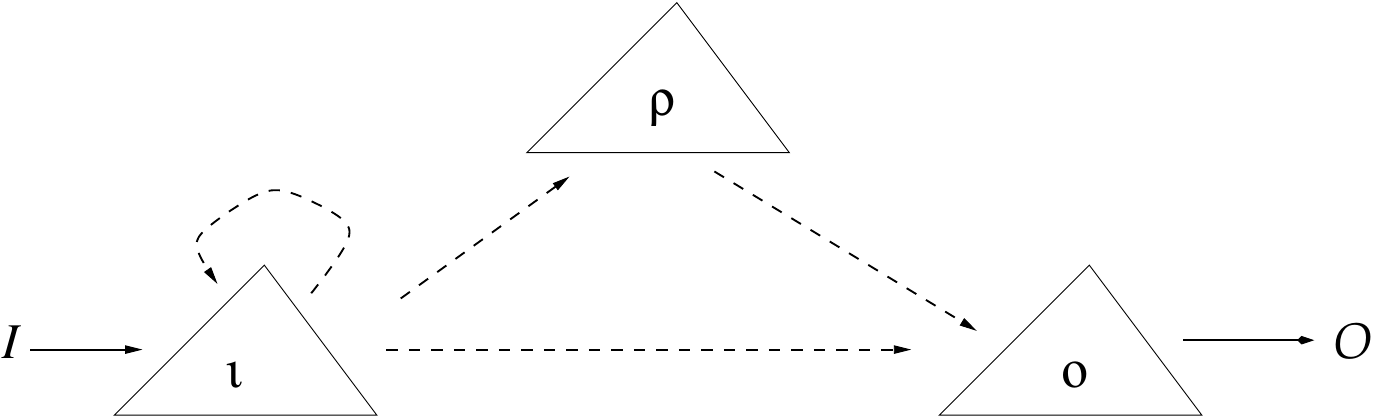}
\caption{GRN}
\label{F:3node_grn2}
\end{subfigure} \qquad
\centering
\begin{subfigure}[b]{0.45\textwidth}
\centering
\includegraphics[width=\textwidth]{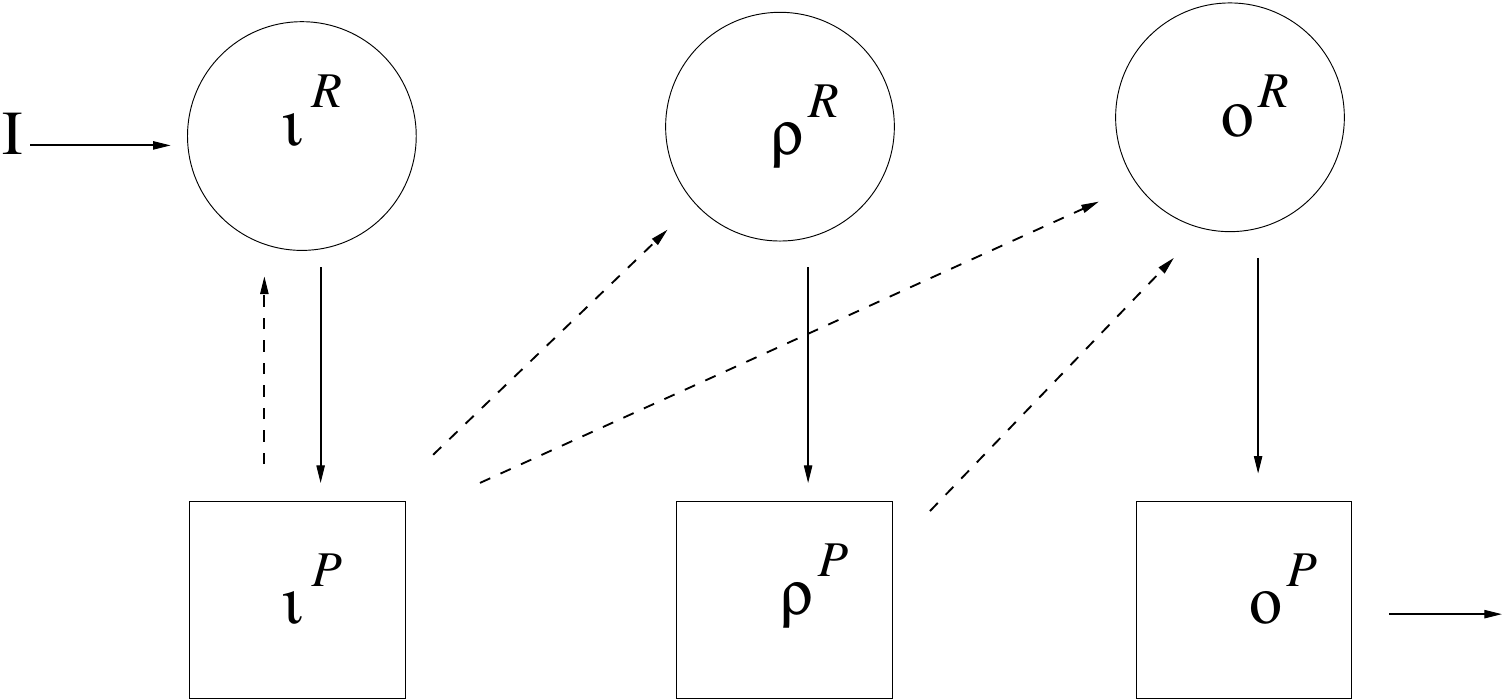}
\caption{PRN}
\label{F:3node_prn2}
\end{subfigure} \qquad
\caption{\label{F:reg_net_1b}
\textbf{Feedforward loop.}
(a) The $3$-node input-output GRN. Triangles designate genes and dashed arrows designate either gene coupling or auto-regulation. 
(b) The corresponding $6$-node PRN. Circles designate mRNA concentrations and squares designate protein concentrations. 
Solid lines stand for ${}^R\longrightarrow {}^P$ coupling inside a single gene and dashed lines for ${}^P\dashrightarrow {}^R$ coupling between genes, that is, the couplings (arrows) inherited from the GRN.}
\end{figure}

Following \cite{GS23}, it is straightforward to associate a class of ODEs (vector fields) to a PRN.
The class of ODEs that are compatible with the network structure are called {\em admissible} (see Section \ref{SS:io network}).
To simplify notation we use the name of each node to refer to the one-dimensional coordinate  corresponding to that node.

\begin{example} \label{E:FFL} \normalfont
Consider the $3$-node feedforward loop GRN shown in Figure~\ref{F:3node_grn2}. 
Its associated PRN is the $6$-node network shown in Figure~\ref{F:3node_prn2}.
To each node of the PRN corresponds a $1$-dimensional state variable.
Hence the total state of the system is given by a vector $(\iota^R, \iota^P, \rho^R, \rho^P, o^R, o^P)\in\R^6$.
The general admissible system of ODEs is
\begin{equation} \label{e:FFL_PRN}
\begin{split}
\dot{\iota}^R & = f_{\iota^R}(\iota^R, \iota^P, \II) \\
\dot{\iota}^P & = f_{\iota^P}(\iota^R, \iota^P) \\
\dot{\rho}^R & = f_{\rho^R}(\iota^P, \rho^R) \\
\dot{\rho}^P & = f_{\rho^P}(\rho^R, \rho^P) \\
\dot{o}^R & = f_{o^R}(\iota^P, \rho^P, o^R) \\
\dot{o}^P & = f_{o^P}(o^R, o^P) \\
\end{split}
\end{equation}
The input node represents the mRNA concentration $\iota^R$ and the output node represents the protein concentration $o^P$. 
The input parameter $\II$ appears explicitly only in the equation of the input node.
\END
\end{example}

In general, the special bipartite structure of the PRN imposes restrictions on the functional form of the admissible vector fields.
For each gene $\rho$ there is a pair of 
of PRN nodes $\rho^R$, $\rho^P$ that yields a pair of $1$-dimensional state variables and corresponding differential equations 
\begin{equation} \label{EQ:GEN_FORM_ADM}
\begin{split}
\dot{\rho}^R & = f_{\rho^R}(\rho^R,\rho^P,\tau^P_1,\ldots,\tau^P_k) \\
\dot{\rho}^P & = f_{\rho^P}(\rho^R, \rho^P) \\
\end{split}    
\end{equation}
Here, $f_{\rho^R}$ and $f_{\rho^P}$ are smooth functions.
The variables $\tau^P_1,\ldots,\tau^P_k$ are the transcription factors (TFs), that is, the corresponding protein concentrations of the genes $\tau_1,\ldots,\tau_k$ that regulate gene $\rho$.
They are determined by the GRN arrows $\tau_i\to\rho$ and the corresponding PRN arrows $\tau^P_i\to\rho^R$.
The presence of the variable $\rho^P$ in the function $f_{\rho^R}$ occurs if and only if gene $\rho$ has a self-coupling in the GRN (see Remark \ref{RMK:self-coupling} below).
If $\rho$ is the input node then the function $f_{\rho^R}$ depends explicitly on the input parameter $\II$, as well.
From now on we will assume only the general form \eqref{EQ:GEN_FORM_ADM} for each protein-coding gene, since our classification results depend only on the GRN coupling structure, not on the particular form of the equations (see Remark \ref{RMK:act_rep}).

\begin{remark}[{\bf Activation and Repression}] \normalfont \label{RMK:act_rep}
Very often, GRNs in the literature are drawn with two types of arrows:
(i) $\tau\rightarrow\rho$ to indicate that gene $\tau$ (more specifically, its protein) acts as an {\em activator}, or {\em excitatory} transcription factor, of gene $\rho$, and (ii) $\tau\flatrightarrow\rho$ to indicate that gene $\tau$ (more specifically, its protein) acts as a {\em repressor}, or {\em inhibitory} transcription factor, of gene $\rho$.
In terms of the associated differential equations this information is encoded in the $\tau$-dependence of the function $f_{\rho^R}$ in \eqref{EQ:GEN_FORM_ADM}.
Typically, the dependence of $f_{\rho^R}$ on $\tau$ is defined by the so called {\em gene input function}. Well-known examples of gene input functions are the classical {\em Michalis-Menten} and {\em Hill} functions \cite{S08}, and their multi-variate versions \cite{KBZDA08}.
Since we do not specify the functional form of the differential equations, we will not use distinct arrow types in the GRN to indicate activation/repression of genes.
However, we do employ distinct arrow types (see \cite{CH10}) in the PRN to distinguish mRNA to protein (${}^R\longrightarrow {}^P$)  and protein to mRNA (${}^P\dashrightarrow {}^R$) couplings.
\END
\end{remark}

\begin{remark}[{\bf Autoregulation}] \normalfont \label{RMK:self-coupling}
Self-coupling, or autoregulation, is a peculiar feature of GRNs.
It means that the gene product acts as a transcription factor of the gene itself.
The dynamical interpretation of autoregulation is revealed by the associated PRN.
It is a coupling from the protein node to the mRNA node of the {\em same} gene (see Figure \ref{F:reg_net_1b}).
Moreover, it is clear from \eqref{e:FFL_PRN} that the self-coupling representing autoregulation is not the same as `self-interaction'.
In fact, all nodes of the PRN are {\em self-interacting}, in the sense that the right-hand side of each differential equation explicitly depends on the state variable on the left-hand side.
The clarification of the dynamical interpretation of autoregulation is another advantage of the PRN formalism.
\END
\end{remark}

\subsection{Infinitesimal Homeostasis in PRN}
\label{SS:io network}

Input-output networks occur naturally when studying homeostasis in biochemical networks \cite{RBGSN17,GS17}.
We recall the setup of an abstract input-output network introduced in \cite{WHAG21}.

An input-output network $\GG$ is a directed graph consisting of $n+2$ nodes. 
There is an {\em input node} $\iota$, an {\em output node} $o$, and $n$ {\em regulatory nodes} $\rho = (\rho_1,\ldots,\rho_n)$.
We assume that every node lies on a path from the input node $\iota$ to the output node $o$. 
That is, the network is a {\em core} network.
An {\em admissible system} associated with the input-output network $\GG$ is a parameterized system of ODEs
\begin{equation}\label{eq:ad io}
\dot{X} = F(X, \II)
\end{equation}
where $X = (x_\iota,x_\rho,x_o) \in \R^{n+2}$ are the node state variables, $\II\in\R$ is the {\em external input parameter}, and $F = (f_\iota,f_{\rho},f_o)$ is the associated vector field.
Explicitly, \eqref{eq:ad io} is the system 
\begin{equation} \label{eq: io}
\begin{split}
\dot{x}_\iota & = 
f_\iota(x_\iota,x_\rho, x_o,\II) \\
\dot{x}_\rho & =
f_\rho(x_\iota,x_\rho,x_o) \\
\dot{x}_o & = f_o(x_\iota,x_\rho,x_o)
\end{split}
\end{equation}
The compatibility of $F$ with the network $\GG$ is given by the following conditions:
\begin{enumerate}[(a)]
\item $f_j$ depends on node $\ell$ only if there is an arrow in the network $\GG$ from $\ell\to j$. 
\item $f_\iota$ is the only vector field component that depends explicitly on $\II$ and $f_{\iota,\II}\neq 0$ generically.
\end{enumerate}
We write $f_{i,j}$ to denote the partial derivative of $f_i$ with respect to $j$ at $(X_0,\II_0)$.

In order to define the notion of `infinitesimal homeostasis' in the context of input-output networks, assume that $X_0$ is a linearly stable equilibrium of \eqref{eq: io} at $\II=\II_0$.
Stability of $X_0$ implies that there is a unique stable equilibrium at $X(\II)=\big(x_\iota(\II),x_\rho(\II),x_o(\II)\big)$ as $\II$ varies on neighborhood of $\II_0$.

\begin{definition} \normalfont
The {\em input-output} function of system \eqref{eq: io}, at the family of equilibria $\big(X(\II),\II\big)$, is the function $\II\to o(\II)$, that is, the projection of $X(\II)$ onto the coordinate $o$.
We say that the input-output function $o(\II)$ exhibits \emph{infinitesimal homeostasis} at $\mathcal{I}_0$, if
\begin{equation} \label{def:inf_homeo_io}
o^\prime(\mathcal{I}_0) = 0
\end{equation}
where $'$ indicates differentiation with respect to $\II$. 
\END
\end{definition}

A straightforward application of Cramer's rule in \cite{WHAG21} gives a formula for determining 
infinitesimal homeostasis.  Let $J$ be the $(n+2)\times (n+2)$ Jacobian matrix of \eqref{eq: io} 
at the equilibrium $X_0$ with $\II=\II_0$,
\begin{equation} \label{jacobian io}
J = \begin{bmatrix}
  f_{\iota, \iota}  &  f_{\iota, \rho} & f_{\iota, o} \\
  f_{\rho,\iota} &  f_{\rho,\rho} & f_{\rho,o} \\
  f_{o,\iota} &  f_{o, \rho} & f_{o, o}
\end{bmatrix}
\end{equation}
The $(n+1)\times (n+1)$ {\em homeostasis matrix} $H$ is obtained from the Jacobian matrix
$J$ by eliminating the first row and the last column, that is,
\begin{equation} \label{H_def}
H = \Matrixc{ f_{\rho,\iota} &  f_{\rho,\rho}\\
  f_{o,\iota} &  f_{o,\rho}}
\end{equation}

\begin{lemma}[{\cite[Lemma 1.5]{WHAG21}}] \label{cramer_rule}
Suppose $o(\mathcal{I})$ is the input-output function of an input-output network $\mathcal{G}$.
Then $\mathcal{I}_0$ is a point of infinitesimal homeostasis if and only if
\begin{equation} \notag
\det(H) = 0
\end{equation}
at the equilibrium $\big(X(\mathcal{I}_0),\mathcal{I}_0 \big)$.
\end{lemma}

\setcounter{theorem}{0}
\begin{excont}[{\bf Continued}] \normalfont
Consider the $6$-node PRN shown in Figure~\ref{F:3node_prn2}.
The input-output function is given by $\II\mapsto o^P(\II)$.
The homeostasis matrix for the corresponding system of ODEs \eqref{e:FFL_PRN} is the $5\times 5$ matrix obtained from the $6\times 6$ Jacobian matrix by deleting its first row and its last column, namely,
\[
H = \Matrixc{ 
f_{\iota^P, \iota^R} & f_{\iota^P,\iota^P} & 0 & 0 & 0 \\
0 & f_{\rho^R,\iota^P} & f_{\rho^R, \rho^R} & 0 &  0 \\
0 & 0 & f_{\rho^P,\rho^R} & f_{\rho^P,\rho^P} & 0 \\
0 & f_{o^R,\iota^P} & 0 & f_{o^R,\rho^P} & f_{o^R,o^R} \\
0 & 0 & 0 & 0 &  f_{o^P,o^R} 
}
\]  
It follows that 
\begin{equation} \label{det_homeo_feedforward}
\det(H) = f_{\iota^P,\iota^R} \,
\big(f_{\rho^R,\iota^P} \, f_{\rho^P,\rho^R} \, f_{o^R,\rho^P} + f_{\rho^R,\rho^R} \, f_{\rho^P,\rho^P} \, f_{o^R,\iota^P} \big)
\, f_{o^P,o^R}
\end{equation}
In this example the decomposition of $\det(H)$ into irreducible factors leads to three different `homeostasis types', two degree $1$ factors and one degree $3$ factor.
The `homeostasis class' associated with the three irreducible factors in this example is called {\em structural homeostasis}.
Structural homeostasis associated to a degree $1$ factor is called {\em Haldane homeostasis}.
\END
\end{excont}
\setcounter{theorem}{2}

Lemma~\ref{cramer_rule} reduces the computation of infinitesimal homeostasis to solving $\det(H) = 0$, where $H$ is the homeostasis matrix in \eqref{H_def}. 
Wang \etal~\cite{WHAG21} use Frobenius-König theory to show the existence of two $(n + 1) \times (n + 1)$ permutation matrices $P$ and $Q$ such that
\begin{equation} \label{phq io}
PHQ = \Matrixc{B_1 & \ast & \cdots & \ast \\
0 & B_2 & \cdots & \ast \\
\vdots &  &  & \vdots \\
0 & 0 & \cdots & B_m}
\end{equation}
where $B_1, \ldots, B_m$ are $m$ unique irreducible square blocks.  Hence
\[
\det(H) = \det(B_1) \cdots \det(B_m)
\]
and no further factorization of $\det(H)$ is possible. 
The unique irreducible square blocks $B_{\eta}$ in \eqref{phq io} are called {\em irreducible blocks}.
The irreducible blocks of a homeostasis matrix $H$ correspond to the possible {\em homeostasis types} that can occur in a network with homeostasis matrix $H$.
Furthermore, they can be divided into exactly two {\em homeostasis classes}: {\em structural} and {\em appendage} (see Appendix \ref{SS:core_network}). 

We say that infinitesimal homeostasis is caused by (or is of type) $B_{\eta}$, if 
\begin{equation} \label{eq:homeostasis_type}
\det(B_{\eta}) = 0
\qquad\text{and}\qquad 
\det(B_{\xi}) \neq 0
\quad\text{for all}\;\xi \neq \eta
\end{equation}
Next, we associate a subnetwork $\KK_\eta$ of $\GG$ to each irreducible block $B_\eta$, called {\em homeostasis subnetwork}. 
The homeostasis subnetworks $\KK_\eta$ are completely determined by the irreducible blocks $B_\eta$, and vice-versa.
Hence, they can be divided into two classes (structural or appendage) according to the homeostasis class of the corresponding irreducible block.
Moreover, the homeostasis subnetworks can be fully characterized in terms of combinatorial properties, so that we can obtain all homeostasis types directly from the network $\GG$.

\setcounter{theorem}{0}
\begin{excont}[{\bf Continued}] \normalfont
The three homeostasis subnetworks of the $6$-node PRN shown in Figure~\ref{F:3node_prn2} can be obtained from the input-output network as explained in Appendix \ref{SS:combinatorial}.
The two one-dimensional structural blocks $[f_{\iota^P,\iota^R}]$ and $[f_{o^P,o^R}]$ correspond to Haldane subnetworks $\iota^R \to \iota^P$ and
$o^R \to o^P$, respectively.
The $3$-dimensional structural block
\begin{equation} \label{E:3D_SB}
\begin{bmatrix}
f_{\rho^R, \iota^P} & f_{\rho^R, \rho^R} & 0 \\
0 & f_{\rho^P, \rho^R} & f_{\rho^P, \rho^P} \\
f_{o^R, \iota^P} & 0 & f_{o^R, \rho^P} 
\end{bmatrix}
\end{equation}
corresponds to the $4$-node subnetwork generated by the nodes 
$\{\iota^P,\rho^R,\rho^P,o^R\}$.
\END
\end{excont}
\setcounter{theorem}{2}

Next, we consider the notion of `homeostasis pattern', namely, the set of nodes that are simultaneously homeostatic whenever the output node is homeostatic.

\begin{definition} \normalfont
Let $\GG$ be an input-output network and suppose that infinitesimal homeostasis occurs in $\GG$ at $\II_0$, that is, $o'(\II_0)=0$.
A {\em homeostasis pattern} of $\GG$ is a set of nodes $\sigma$ (including the output node $o$) for which the condition $\sigma^\prime(\II_0) = 0$, holds generically.
\END
\end{definition}

Homeostasis patterns can be graphically represented by coloring the nodes of $\GG$ that are homeostatic.
It turns out that each homeostasis pattern is fully determined by a homeostasis subnetwork and vice-versa.
In other words, there is complete correspondence between the homeostasis types, the homeostasis subnetworks and the homeostasis patterns of $\GG$ (see Appendix \ref{SS:homeo_induce}).

\setcounter{theorem}{0}
\begin{excont}[{\bf Continued}] \normalfont
The three homeostasis subnetworks of the $6$-node PRN shown in Figure~\ref{F:reg_net_1b}(b) lead to three distinct homeostasis patterns shown in Table \ref{T:6nodePRN}.
Notice that the two Haldane homeostasis subnetworks correspond to the coupling between the mRNA and protein in the {\em same} gene.
This creates a homeostasis pattern where one gene has
one homeostatic PRN-node while the other PRN-node is not.
Whereas in the pattern associated to the $3$-dimensional structural type both associated PRN-nodes of each gene are either simultaneously homeostatic or not.
Furthermore, it can be shown that the $3$-dimensional structural type is the cause of homeostasis in the example of \cite{AGS18}.
This is most easily seen by comparing the homeostasis patterns: it is shown in \cite[Fig. 3]{AGS18} that only the PRN-nodes associated to the output node of the GRN are homeostatic, which corresponds exactly to the homeostasis pattern $\{o^R,o^P\}$ in the PRN in Figure~\ref{F:3node_prn2}.
\END
\end{excont}
\setcounter{theorem}{2}

\begin{table}[!htb]
\renewcommand{\arraystretch}{1.5}
\begin{equation*}
\begin{array}{ccc} \hline
\text{Homeostasis Type} & \text{Homeostasis Subnetwork}  & \text{Homeostasis Pattern} \\ \hline
{[f_{\iota^P,\iota^R}]} & \iota^R \to \iota^P
& \{\iota^P, \rho^R, \rho^P, o^R, o^P\} \\ 
{[f_{o^P,o^R}]} & o^R \to o^P
& \{o^P\} \\ 
\text{$3$D matrix \eqref{E:3D_SB}} 
& \big\langle\iota^P,\rho^R,\rho^P,o^R\big\rangle
&  \{o^R, o^P\} \\ \hline
\end{array}
\end{equation*}
\caption{\label{T:6nodePRN}
Infinitesimal homeostasis and corresponding homeostasis patterns of the $6$-node PRN shown in Figure~\ref{F:3node_prn2}.
In the second column the notation $\langle\,\cdot\,\rangle$ stands for the `subnetwork generated'.
In the last column we list the homeostatic nodes.
}
\end{table}

\subsection{Infinitesimal Homeostasis in GRN}

As we have shown using the GRN of Figure \ref{F:reg_net_1a} 
(or its abstract version in Figure \ref{F:3node_grn2} the PRN construction allows us to apply the theory of \cite{WHAG21} to obtain all the possible homeostasis types and corresponding homeostasis patterns on the PRN.
Now it remains to explain how the results obtained for the PRN can be `lifted back' to the GRN.
In other words, we first need to define what it means for a node in a GRN to be homeostatic.
Then, we use the general theory to determine in a purely combinatorial fashion, the `formal homeostasis subnetworks' of the GRN
(see Appendix \ref{SS:combinatorial}).
Lastly, we show how the homeostasis subnetworks of the GRN and the associated PRN relate to each other.
Since only the homeostasis subnetworks of the PRN have a `dynamical' interpretation of infinitesimal homeostasis, we use the relation obtained before reinterpreting infinitesimal homeostasis on the GRN level.

In order to explain this relation we will consider another simple Example \ref{E:FBI}, in addition to Example \ref{E:FFL}.
Figure \ref{F:3node_grn} shows a $3$-node GRN called {\em feedback inhibition}, which plays an important role in the GRN of {\em E. coli} and Figure \ref{F:3node_prn} is the associated PRN.
Feedback inhibition appears twice in the regulatory cascade of carbohydrate catabolism of {\em E. coli} \cite[Figs. 1 and 2(a)]{MJT08}, with $(\iota,\tau,o)=(\textrm{IHF},\textrm{CRP},\textrm{FIS})$ and $(\iota,\tau,o)=(\textrm{ARC-A},\textrm{HNS},\textrm{GAD-X})$.

\begin{figure}[!htb]
\begin{subfigure}[b]{0.3\textwidth}
\centering
\includegraphics[width=\textwidth]{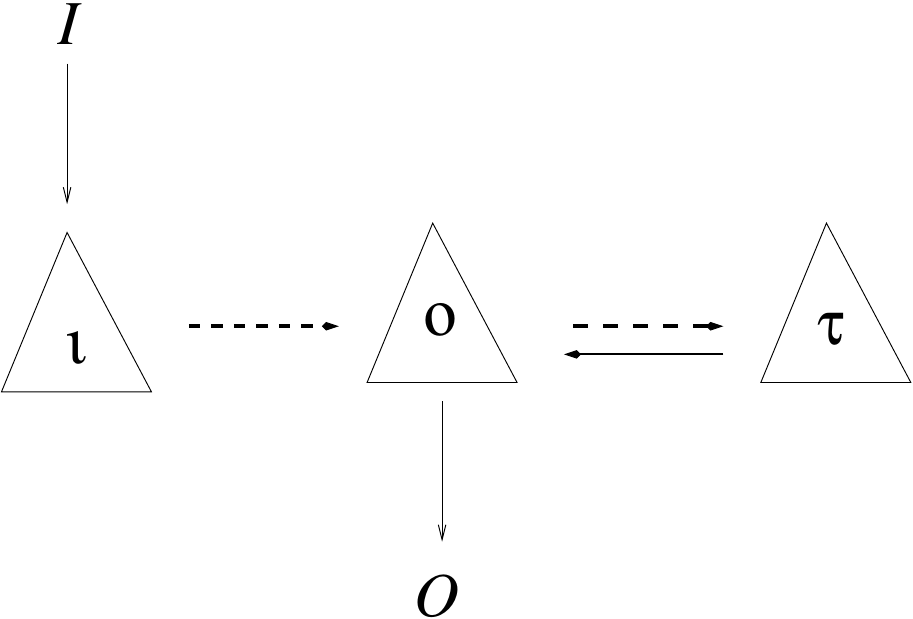}
\caption{GRN}
\label{F:3node_grn}
\end{subfigure} \qquad\qquad\qquad
\centering
\begin{subfigure}[b]{0.25\textwidth}
\centering
\includegraphics[width=\textwidth]{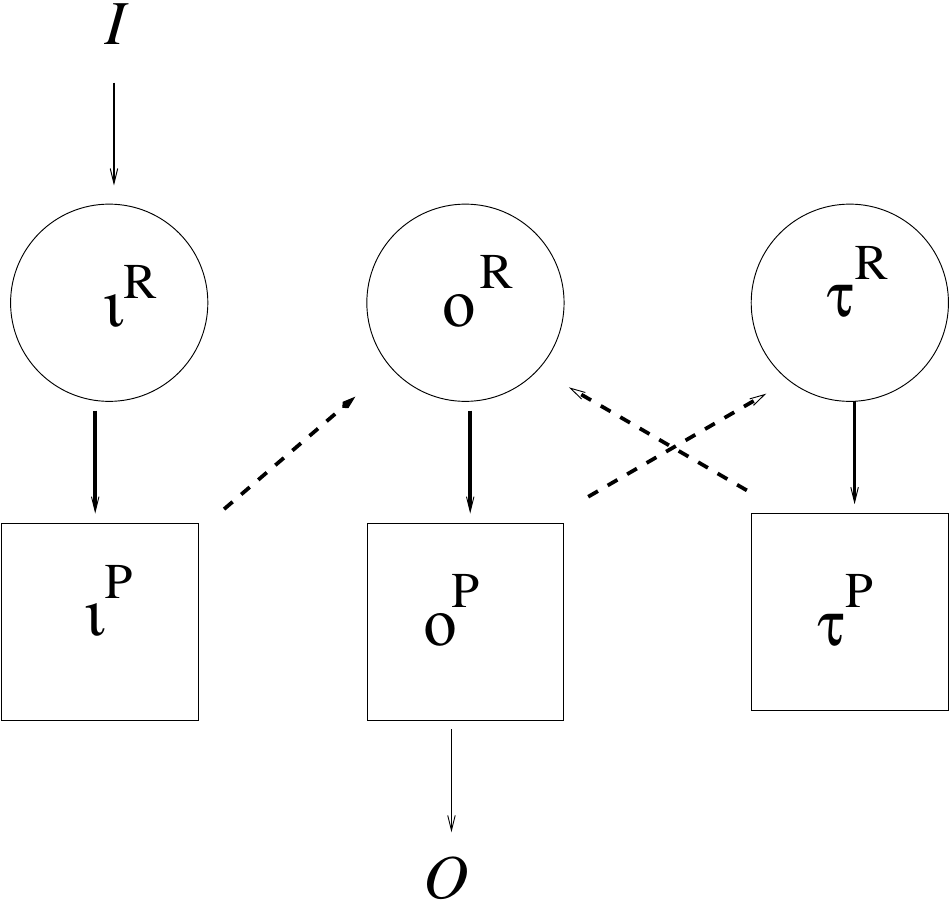}
\caption{PRN}
\label{F:3node_prn}
\end{subfigure} \qquad
\caption{\label{F:3node}
\textbf{Feedback inhibition.}
(a) The $3$-gene input-output GRN.
Triangles designate genes
and dashed arrows designate gene coupling.
(b) The corresponding $6$-node input-output PRN.
Circles designate mRNA concentrations and squares designate protein concentrations. 
Solid lines ${}^R\longrightarrow {}^P$ stand for coupling inside a single gene and dashed lines ${}^P\dashrightarrow {}^R$ for coupling between genes, that is, the couplings inherited from the GRN.}
\end{figure}

\begin{example} \normalfont \label{E:FBI}
Consider the $3$-node feedback inhibition GRN, Figure~\ref{F:3node_grn}, and the associated $6$-node PRN in Figure~\ref{F:3node_prn}.
The total state is given by a vector $(\iota^R, \iota^P, \rho^R, \rho^P, o^R, o^P)\in\R^6$ and general admissible system of ODEs is
\begin{equation} \notag
\begin{split}
\dot{\iota}^R & = f_{\iota^R}(\iota^R,\II) \\
\dot{\iota}^P & = f_{\iota^P}(\iota^R,\iota^P) \\
\dot{\tau}^R & = f_{\tau^R}(\tau^R,o^P) \\
\dot{\tau}^P & = f_{\tau^P}(\tau^R,\tau^P) \\
\dot{o}^R & = f_{o^R}(\iota^P,\tau^P,o^R) \\
\dot{o}^P & = f_{o^P}(o^R,o^P)
\end{split}
\end{equation}
The homeostasis matrix is
\begin{equation} \notag
H = 
\begin{bmatrix}
f_{\iota^P, \iota^R} & f_{\iota^P, \iota^P} & 0 & 0 & 0 \\
0 & 0 & f_{\tau^R, \tau^R} & 0 & 0 \\
0 & 0 & f_{\tau^P, \tau^R} & f_{\tau^P, \tau^P} & 0 \\
0 & f_{o^R, \iota^P} & 0 & f_{o^R, \tau^P} & f_{o^R, o^R} \\
0 & 0 & 0 & 0 & f_{o^P, o^R}
\end{bmatrix}
\end{equation}
and the homeostasis determinant is
\begin{equation} \label{det_homeo_feedback}
\det(H) = f_{\tau^R,\tau^R} f_{\tau^P,\tau^P} f_{o^R,\iota^P} f_{\iota^P,\iota^R} f_{o^P,o^R}
\end{equation}
The homeostasis determinant is a completely reducible polynomial of degree five.
Thus the PRN has five degree $1$ homeostasis types.
The associated homeostasis subnetworks and the homeostasis patterns are shown in Table \ref{T:3node_PRN}.
\END
\end{example}

\begin{table}[!htb]
\renewcommand{\arraystretch}{1.5}
\begin{equation*}
\begin{array}{ccc} \hline
\text{Homeostasis Type} & \text{Homeostasis Subnetwork}  & \text{Homeostasis Pattern} \\ \hline
{[f_{\tau^R,\tau^R}]} & \tau^R & \{o^R, o^P\} \\
{[f_{\tau^P,\tau^P}]} & \tau^P & \{\tau^R, o^R, o^P\} \\
{[f_{o^R,\iota^P}]} & \iota^P \to o^R & \{\tau^R, \tau^P, o^R, o^P\} \\
{[f_{\iota^P,\iota^R}]} & \iota^R \to \iota^P & \{\iota^P, \tau^R, \tau^P,  o^R, o^P\} \\
{[f_{o^P,o^R}]} & o^R \to o^P & \{\tau^R, \tau^P,  o^P\} \\ \hline
\end{array}
\end{equation*}
\caption{\label{T:3node_PRN}
Infinitesimal homeostasis and corresponding homeostasis patterns of the $6$-node PRN shown in Figure \ref{F:3node_prn}.
In the last column we list the homeostatic nodes.}
\end{table}

For comparison, let us determine the homeostasis subnetworks and homeostasis patterns obtained by working directly with the GRN.
Recall that each homeostasis subnetwork corresponds to a `homeostasis mechanism' that can cause homeostasis and they can be divided into two classes: (i) {\em structural class}, corresponds to `generalized feedforward mechanisms' and (ii) {\em appendage class} corresponds to  `generalized feedback mechanisms'.

Let us start with Example \ref{E:FFL}, Figure \ref{F:3node_grn2}.
First, observe that the self-coupling of the input node $\iota$ does not affect the construction of homeostasis subnetworks (we simply keep the arrow during the procedure).
Then, since $\iota$ and $o$ are the only super-simple nodes and there are no appendage nodes, it follows that there is only one homeostasis subnetwork of structural class, see Table \ref{T:GRN_FFL}.
Comparing with Table \ref{T:6nodePRN} we see that the associated PRN, Figure \ref{F:3node_prn2}, has three homeostasis subnetworks.
The two first homeostasis subnetworks of the PRN are Haldane subnetworks.
Moreover, in both cases the two PRN-nodes of the Haldane subnetwork come from the same GRN-node.
This observation suggests that these two Haldane subnetworks do not have a `counterpart' in the GRN.

\begin{table}[!htb]
\renewcommand{\arraystretch}{1.5}
\begin{equation*}
\begin{array}{ccc} \hline
\text{Homeostasis Class} & \text{Homeostasis Subnetwork}  & \text{Homeostasis Pattern} \\ \hline
\text{Structural} 
& \big\langle\iota,\rho,o\big\rangle
&  \{o\} \\ \hline
\end{array}
\end{equation*}
\caption{\label{T:GRN_FFL}
Infinitesimal homeostasis and corresponding homeostasis patterns of the GRN shown in Figure \ref{F:3node_grn2}.
In the second column the notation $\langle\,\cdot\,\rangle$ stands for the `subnetwork generated'.
In the last column we list the homeostatic nodes.}
\end{table}

Now we consider Example \ref{E:FBI}, Figure \ref{F:3node_grn}.
Here, we have that the only super simple nodes are $\iota$ and $o$, but the regulatory node $\tau$ is appendage.
Then, there are two homeostasis subnetworks, one of structural class and the other of appendage class, see Table \ref{T:GRN_FBI}.
Comparing with Table \ref{T:3node_PRN} we see that the associated PRN, Figure \ref{F:3node_prn}, has five homeostasis subnetworks (two appendage and three structural).
Among the three structural subnetworks of the PRN, two of them (the last two cases in Table \ref{T:3node_PRN}) are such that the two PRN-nodes come from the same GRN-node.
The third case in Table \ref{T:3node_PRN}) is different since the PRN-nodes come from distinct GRN-nodes.
The two appendage subnetworks of the PRN are given by a single appendage node (i.e., a degree $1$ irreducible factor).
Appendage homeostasis associated to degree $1$ irreducible factors is called {\em null-degradation homeostasis}.

\begin{table}[H]
\renewcommand{\arraystretch}{1.5}
\begin{equation*}
\begin{array}{ccc} \hline
\text{Homeostasis Class} & \text{Homeostasis Subnetwork}  & \text{Homeostasis Pattern} \\ \hline
\text{Appendage} & \tau & \{o\} \\
\text{Structural} & \iota \to o & \{\tau, o\} \\ \hline
\end{array}
\end{equation*}
\caption{\label{T:GRN_FBI}
Infinitesimal homeostasis and corresponding homeostasis patterns of the GRN shown in Figure \ref{F:3node_grn}.
In the last column we list the homeostatic nodes.}
\end{table}

\noindent
These examples suggest the following observations.
\begin{enumerate}[(1)]
\item
The simplest example of structural homeostasis is {\em Haldane homeostasis}.
PRNs can have two types of Haldane homeostasis: (i) one corresponds to an arrow connecting an mRNA node to a protein node (e.g., the first two cases in Example \ref{T:6nodePRN} and the two last cases in Table \ref{T:3node_PRN} and (ii) the other corresponds to an arrow connecting a protein node to an mRNA node (e.g., the third case in Table \ref{T:3node_PRN}).

\item The simplest example of appendage homeostasis is {\em null-degradation homeostasis}.
In general, there are two different types of null-degradation in PRNs: (i) one occurs in a protein node (e.g., the second case in Table \ref{T:3node_PRN}) and (ii) the other occurs in an mRNA node (e.g., the first case in Table \ref{T:3node_PRN}).
\end{enumerate}

\begin{definition} \normalfont
\label{def:special_component}
Let $\GG$ be a GRN with associated PRN $\RR$.
Let $\tau\in\GG$ be a node with
$\tau^R\in\RR$ be the mRNA node and $\tau^P\in\RR$ be the protein node.
\begin{enumerate}[(a)]
\item An appendage node $\tau\in\GG$ is a {\em single appendage node} if $\{\tau\}$ is an appendage subnetwork of $\GG$ with no self-coupling.
\item If $\{\tau^R\}$ is an appendage subnetwork of $\RR$ then it is called {\em $R$-null-degradation}.
\item If $\{\tau^P\}$ is an appendage subnetwork of $\RR$ then it is called {\em $P$-null-degradation}.
\item If $\langle\tau^R,\tau^P\rangle$ is a structural subnetwork of $\RR$ then it is called {\em $\RR$-Haldane}.
\END
\end{enumerate}
\end{definition}

Now we state the first main result of the paper, regarding the relation between the homeostasis subnetworks of a GRN and its associated PRN (see Appendix \ref{SS:combinatorial} for the terminology and Appendix \ref{SS:homeo subnetworks in RPN} for the proof.)

\begin{theorem} \label{thm:main1}
The homeostasis subnetworks of a GRN $\GG$ and its associated PRN $\RR$ correspond uniquely to each other, except in the following cases:
\begin{enumerate}[{\rm (a)}]
\item The $\RR$-Haldane subnetworks of $\RR$ correspond uniquely to the super-simple nodes of $\GG$ (super-simple nodes are not homeostasis subnetworks of $\GG$).
\item For every single appendage node $\tau\in\GG$ the appendage subnetwork $\{\tau\}$ of $\GG$ yields 2 appendage subnetworks $\{\tau^R\}$ ($R$-null-degradation) and $\{\tau^P\}$ ($P$-null-degradation) of $\RR$.
\end{enumerate}
\end{theorem}
\begin{proof}
It follows from Theorem \ref{thm: homeo subnetwork in PRN} from Appendix \ref{SS:homeo subnetworks in RPN}.
\end{proof}

Now we consider the following question: {\em Can we combinatorially determine the homeostasis patterns of a GRN and its associated PRN?}
In principle this can be done since the classification of homeostasis patterns is purely combinatorial.
We start by introducing the notion of a homeostasis pattern in a GRN that is derived from the associated PRN.

\begin{definition} \normalfont
\label{def: pattern in GRN}
Consider a GRN $\GG$ and its associated PRN $\RR$.
Suppose that infinitesimal homeostasis occurs in the PRN at $\II_0$.
A node $\rho\in\GG$ is said to be {\em GRN-homeostatic} if both associated PRN-nodes $\rho^R$ and $\rho^P$ are simultaneously homeostatic at $\II_0$.
A \emph{GRN-generating homeostasis pattern} is a homeostasis pattern $\PP$ on $\RR$ such that, for every PRN-node in $\PP$ corresponds to GRN-homeostatic node.
\END
\end{definition}

In a GRN-generating homeostasis pattern all PRN-nodes appear in mRNA-protein pairs.
That is, the set of GRN-homeostatic nodes match perfectly the PRN homeostasis pattern.

Let us consider again our two running examples.
The results for the feedforward loop GRN and its associated PRN, Figure \ref{F:reg_net_1b}, are shown in Table \ref{T:FFL_GRN_PRN} and for the feedback inhibition GRN and the associated PRN, Figure \ref{F:3node}, are shown in Table \ref{T:FBI_GRN_PRN}.

\begin{table}[!htb]
\renewcommand{\arraystretch}{1.5}
\begin{equation*}
\begin{array}{ccccc} \hline
\text{PRN Subnet} & \text{GRN Subnet} & \text{PRN Pattern} & \text{GRN-H Nodes} & \text{GRN Pattern} \\ \hline
\iota^R \to \iota^P & \iota^* & \{\iota^P, \rho^R, \rho^P, o^R, o^P\} & \{\rho,o\} & - \\ 
o^R \to o^P & o^* & \{o^P\} & - & - \\ 
\big\langle\iota^P,\rho^R,\rho^P,o^R\big\rangle
& \big\langle\iota,\rho,o\big\rangle & \{o^R, o^P\} & \{o\} & \{o\} \\ \hline
\end{array}
\end{equation*}
\caption{\label{T:FFL_GRN_PRN}
Infinitesimal homeostasis and homeostasis patterns of the feedforward loop network GRN and its associated PRN from Figure~\ref{F:reg_net_1b}.
The ${}^*$ denotes the GRN super-simple node corresponding to the $\RR$-Haldane subnetwork.
See Appendix \ref{SS:algorithm} for the computations.
}
\end{table}

\begin{table}[!htb]
\renewcommand{\arraystretch}{1.5}
\begin{equation*}
\begin{array}{ccccc} \hline
\text{PRN Subnet} & \text{GRN Subnet} & \text{PRN Pattern} & \text{GRN-H Nodes} & \text{GRN Pattern} \\ \hline
\tau^R & \tau & \{o^R, o^P\} & \{o\} & \{o\} \\
\tau^P & \tau & \{\tau^R, o^R, o^P\} & \{o\} & - \\
\iota^P \to o^R & \iota \to o & \{\tau^R, \tau^P, o^R, o^P\} & \{\tau,o\} & \{\tau,o\} \\
\iota^R \to \iota^P &  \iota^* & \{\iota^P, \tau^R, \tau^P,  o^R, o^P\} & \{\tau,o\} & - \\
o^R \to o^P & o^* & \{\tau^R, \tau^P, o^P\} & \{\tau\} & - \\ \hline
\end{array}
\end{equation*}
\caption{\label{T:FBI_GRN_PRN}
Infinitesimal homeostasis and homeostasis patterns of the feedforward loop network GRN and its associated PRN from Figure~\ref{F:3node}.
The ${}^*$ denotes the GRN super-simple node corresponding to the $\RR$-Haldane subnetwork.
See Appendix \ref{SS:algorithm} for the computations.
}
\end{table}

\noindent
From these examples we draw a couple of remarks:
\begin{enumerate}[(1)]
\item $\RR$-Haldane subnetworks do not produce GRN-generating homeostasis patterns.
There is always a homeostatic protein node whose pairing mRNA node is not homeostatic.
\item $P$-null-degradation subnetworks do not produce GRN-generating homeostasis patterns.
There is always a homeostatic mRNA node whose pairing protein node is not homeostatic.
\END
\end{enumerate}

These observations hold the essence of the second main result of the paper, regarding the relation between the homeostasis patterns of a GRN and its associated PRN.

It is useful to introduce the following terminology.
Let $\KK$ be a homeostasis subnetwork of the GRN $\GG$.
Then we define a map $\KK\to\KK^R$ from the set of homeostasis subnetworks of $\GG$ to the set of homeostasis subnetworks of the associated $\RR$ as follows.
If $\KK\neq\{\tau\}$, where $\tau$ is a single appendage node, the $\KK^R$ is the unique subnetwork given by Theorem \ref{thm:main1}.
If $\KK=\{\tau\}$, where $\tau$ is a single appendage node, then $\KK^R=\{\tau^R\}$.
With this definition, the map $\KK\to\KK^R$ is injective and the complement of its image in the set of homeostasis subnetworks of the PRN is exactly the set of $\RR$-Haldane and $P$-null-degradation subnetworks.

\begin{theorem} \label{thm:main2}
Let $\GG$ be a GRN and $\RR$ its associated PRN.
Then the homeostasis patterns on $\GG$ correspond exactly to the GRN-generating homeostasis patterns on $\RR$.
The homeostasis patterns on $\RR$ associated to the $\RR$-Haldane and the $P$-null-degradation subnetworks do not correspond to homeostasis patterns on $\GG$.
\end{theorem}
\begin{proof}
This follows from Theorem \ref{thm: patterns of homeostasis in RPN} and Lemma \ref{lem: connection on pattern of homeo in the GRN} from Appendix \ref{SS: homeo inducing in the GRN}.
\end{proof}

\begin{remark}[{\bf Non-coding genes}] \normalfont \label{RMK:non_coding}
As mentioned before, in eukaryotic cells there are several regulatory mechanisms modulating transcription and translation.
Almost all regulatory modulation is performed by non-coding genes, i.e., genes that are transcribed into RNA, but the RNA is not translated into protein.
The PRN formalism can be extended to include non-coding genes thanks to the following observation: the gene product of a non-coding gene is an mRNA whose regulatory activity is performed by direct interaction with other mRNAs
\cite{SFHNABM07,LWV16}.
Thus, unlike a protein-coding gene, a non-coding gene $\nu$ yields only one scalar state variable and one differential equation 
\[
\dot{\nu} = f_{\nu}(\nu, \underbrace{\tau^P_1,\ldots,\tau^P_k}_{\text{TFs}}, \underbrace{\rho^R_1,\ldots,\rho^R_\ell}_{\text{mRNAs}})
\]
Here, $f_{\nu}$ is a smooth function.
The variables $\tau^P_1,\ldots,\tau^P_k$ are the protein concentrations associated to the transcription factors (TFs) that regulate gene $\nu$.
The variables $\rho^R_1,\ldots,\rho^R_\ell$ are the
mRNA concentrations associated to the protein-coding genes that interact with $\nu$.
Finally, for each $\rho^R_j$ above, the corresponding mRNA equation must now depend on $\nu$:
\[
\dot{\rho}^R_j = f_{\rho^R_j} (\rho^R_j,\ldots,\nu)
\]
The consequence for the PRN diagram is that a non-coding gene: (i) gives rise to a single PRN-node $\nu=\nu^R$, instead of two, (ii) receives arrows from protein nodes $\tau^P_i\to\nu^R$, (iii) has a bidirectional connection with the mRNA nodes that it interacts with $\rho^R_j \biarrow \nu$.
With these new requirements the PRN is no longer a bipartite digraph, now it is a {\em tripartite digraph}.
\END
\end{remark}

\section{Conclusion and Outlook}

In this paper, we present a framework for the analysis and classification of homeostasis in gene regulatory networks.
We accomplish this by combining a formalism for the study of gene regulatory networks (GRN), called {\em protein-mRNA networks} (PRN), with
the theories of Wang \etal\cite{WHAG21} and Duncan \etal\cite{DABGNRS23},
for the classification of homeostasis types and homeostasis patterns in input-output networks, respectively.

Given an arbitrary input-output GRN (consisting of protein-coding genes) $\GG$, we associate an input-output PRN $\RR$, which enables us to apply the results of \cite{WHAG21,DABGNRS23}.
By comparing the results for the PRN with a suitable application of the combinatorial piece of the theory to the GRN, we obtain a refinement of the classification of homeostasis types and homeostasis patterns of $\RR$.
The final result is a complete characterization of homeostasis types and homeostasis patterns on the PRN that have correspondent on the GRN.

An interesting byproduct is the discovery 
of homeostasis types and homeostasis patterns on the PRN without GRN counterpart.
The `new' PRN homeostasis types are degree one homeostasis types, namely, they are related to one dimensional irreducible factors of the homeostasis determinant.
They are: (i) {\em $\RR$-Haldane}, that occurs when the linearized coupling between the mRNA and protein of the same gene changes from excitation to inhibition as the input parameter varies, and (ii) {\em $P$-null-degradation}, that occurs when the linearized self-interaction of a protein changes from degradation to production as the input parameter varies.
Although the existence of $\RR$-Haldane and $P$-null-degradation is mathematically established, their occurrence in biological models is unlikely.
$\RR$-Haldane homeostasis is related to the {\em synthesis rate} of the protein from the mRNA template and $P$-null-degradation is related to the {\em degradation rate} of the protein.
Both these rates are {\em constant} (the first is positive and the second is negative) in specific model equations for GRN modeling \cite{PHB09,MSTZ16}.

The main novelties of our approach include: (i) the {\em simultaneous} use of two networks, the GRN and the PRN, in the analysis of gene expression homeostasis, and (ii) the {\em lack of assumptions} about the functional form of the differential equations.

The protein-mRNA formalism is well-known in the literature \cite{PHB09,MSTZ16}.
We take this formalism one step further, by formally introducing the protein-mRNA network (PRN) and using it in conjunction with the GRN to extract a complete view of gene expression homeostasis. See \cite{KHH20} for an approach similar to the construction of the PRN.

Generally speaking, all model equations for gene expression found in the literature have an explicit functional form \cite{AGS18,KEBC05,M05,M08,PHB09,MSTZ16}. 
Here, we assume only the general form \ref{EQ:GEN_FORM_ADM}, forced by the admissibility of vector fields with respect to the PRN.
Hence, our classification results apply to virtually any model equation for gene expression.
Even more importantly, this leaves open the possibility to use `higher-order' terms to model more complicated interactions \cite{BGHS21}.
In the terminology of \cite{GS23} our results are called {\em model independent}.
This means that the classification results obtained here provide a complete list of possible behaviors, with respect to homeostasis, that is {\em independent} of the model equations -- the list depends only on the topology of the network.
Which of those behaviors will be observed in a particular realization of the dynamics (e.g., a model equation) {\em depends} on the specific form of the dynamics.

There are several relevant ways to generalize and extend the theory of homeostasis in gene regulatory networks.
\begin{description}
\item[{\it Multiple inputs.}]
The defining condition for occurrence of homeostasis \eqref{eq:homeostasis_type} is generic for single-variable input-output functions.
The notion of infinitesimal homeostasis can be naturally generalized to multi-variable input-output functions and the theory of homeostasis types can be extended to this setting \cite{GS18,MF21,MF22}.
In this case, we are lead to consider {\em higher codimension homeostasis} \cite{DABGNRS23b}.

\item[{\it Dynamics and bifurcations in PRN.}]
Recall that the starting point of the analysis of homeostasis is to assume that there is a family of asymptotically stable equilibria and define the input-output function from it.
One can take a step back and ask about the existence and uniqueness of equilibria in a given GRN, that is, how the dynamics and bifurcations of a GRN is constrained by the network structure.
In \cite{MMRS23} the authors propose an approach based on PRN, where they analyze the general models, given by the admissible systems, as considered here, and then to specialize the results to specific models.

\item[{\it Additional biological mechanisms.}]
In order to keep the exposition as simple as possible, we have used a `minimal' definition of PRN, consisting only of protein-coding genes and accounting only for the transcription and translation processes. 
Nevertheless, we have hinted that is not difficult to extend the framework to include other biological mechanisms, for instance, non-coding genes (see Remark \ref{RMK:non_coding}.
Furthermore, there are several possible biological mechanisms that can be included in the modeling of gene expression.
Here is a small sample of relevant biological mechanisms: (a) spatial localization of the transcription and translation processes \cite{GK70,CPS15}, (b) transcriptional time delays \cite{M03}, (c) multi-site phosphorylation / dephosphorylation of transcription factors \cite{GA13}, (d) DNA-level transcriptional regulation \cite{DTAC13}.

\item[{\it Relation to other approaches.}]
As we mentioned earlier there is another approach to the modeling of GRN, based on the QSSA \cite{SGT17}.
It is called {\em simplified GRN models} in \cite{PHB09}, {\em protein regulatory networks} in \cite{TB10} and
{\em protein-only models} in \cite{EMMD15}.
Such models can be used only with protein-coding GRN, and provide a perfect match between the nodes of the GRN and the equations for protein concentrations.
A substantial difference between mRNA-protein models and protein-only models is that their {\em generic dynamics} are not the same \cite{EMMD15}; for instance, the former can have oscillatory solutions while the latter can not \cite{MMRS23}.
Finally, as discussed before, the PRN formalism easily allows for extensions and generalizations, by `unfolding' each gene node into two or more PRN-nodes representing the concentrations of different intermediate molecules.  
On the other hand, the simplified GRN formalism reverses this process, by lumping all the intermediate molecular products onto a single (protein) concentration.
\end{description}

\paragraph{Acknowledgments.}
The research of FA was supported by Funda\c{c}\~ao de Amparo \`a Pes\-qui\-sa do Estado de S\~ao Paulo (FAPESP) grants 2019/12247-7 and 2019/21181-0.

\newpage

\section*{Appendix}

\appendix

\section{Homeostasis in Input-Output Networks}
\label{S:Graphic_theory}

In this section we recall basic terminology and results on infinitesimal homeostasis in input-output networks from~\cite{WHAG21,DABGNRS23}.

\subsection{Core Networks and Homeostasis Classes}
\label{SS:core_network}

Wang \etal~\cite{WHAG21} show that the determination of infinitesimal homeostasis in an input-output 
network reduces  to the study of an associated `core subnetwork'.   A \emph{core network} is an input-output network where every node is both upstream from the output node and downstream from the input node.
 
Every input–output network $\mathcal{G}$ has a \emph{core subnetwork} $\mathcal{G}_c$ whose nodes are the nodes in $\mathcal{G}$ that are both upstream from the output node and downstream from the input node and whose arrows are the arrows in $\mathcal{G}$ whose head and tail nodes are both nodes in $\mathcal{G}_c$.
Throughout this paper we assume that input-output networks are core networks.

Suppose that $B_{\eta}$ is an irreducible component in the decomposition \eqref{phq io}, where $B_{\eta}$ is a $k \times k$ diagonal block, that is, $B_{\eta}$ has degree $k$. 
Since the entries of $B_{\eta}$ are entries of $H$, these entries have the form $f_{\rho, \tau}$; that is, the entries are either 0 (if $\tau \to \rho$ is not an arrow in $\GG$), self-coupling (if $\tau = \rho$), or coupling (if $\tau \to \rho$ is an arrow in $\GG$).

\begin{definition} \normalfont
\label{def:irreducible_component}
Let $B_{\eta}$ be an irreducible component in the decomposition \eqref{phq io}.
\begin{enumerate}[(a)]
\item The homeostasis class of type $B_{\eta}$ of degree $k$ is
{\em appendage} if $B_{\eta}$ has $k$ self-couplings and {\em structural} if $B_{\eta}$ has $k-1$ self-couplings.
\item The subnetwork $\KK_{\eta}$ of $\GG$ associated with the homeostasis block $B_{\eta}$ is defined as follows. The nodes in $\KK_{\eta}$ are the union of nodes $p$ and $q$ where $f (p, x_q)$ is a nonzero entry in $B_{\eta}$ and the arrows of $\KK_{\eta}$ are the union of arrows $q \to p$ where $p \neq q$. \END
\end{enumerate}
\end{definition}

The following theorem shows that all irreducible components belong to either appendage or structural class.
When $B_{\eta}$ is appendage, the subnetwork $\KK_{\eta}$ has $k$ nodes and $B_{\eta}$ can be put in a {\em standard Jacobian form}.
Also, when $B_{\eta}$ is structural, the subnetwork $\KK_{\eta}$ has $k + 1$ nodes and $B_{\eta}$ can be put in a {\em standard homeostasis form}.

\begin{theorem}[\cite{WHAG21}] \label{thm: hemeostasis block}
Let $H$ be an $(n + 1) \times (n + 1)$ homeostasis matrix and let $B_{\eta}$ be a $k \times k$ irreducible component in the decomposition \eqref{phq io} with $k \geq 1$. 
Then $B_{\eta}$ has either $k-1$ or $k$ self-coupling.  Furthermore,
\begin{enumerate}[{\rm (a)}]
\item If $B_{\eta}$ has $k-1$ self-coupling entries, then $B_{\eta}$ has the form:
\begin{equation} \notag
\Matrixc{f_{\rho_1, \rho_1} & \cdots & f_{\rho_1, \rho_{k-1}} & f_{\rho_1, l} \\
\vdots & \ddots & \vdots & \vdots \\
f_{\rho_{k-1}, \rho_1} & \cdots & f_{\rho_{k-1}, \rho_{k-1}} & f_{\rho_{k-1}, l} \\
f_{j, \rho_1} & \cdots & f_{j, \rho_{k-1}} & f_{j, l}}
\end{equation}

\item If $B_{\eta}$ has $k$ self-coupling entries, then $B_{\eta}$ has the form:
\begin{equation} \notag
\Matrixc{f_{\rho_1, \rho_1} & \cdots & f_{\rho_1, \rho_k} \\
\vdots & \ddots & \vdots \\
f_{\rho_k, \rho_1} & \cdots & f_{\rho_k, \rho_k}}
\end{equation}
\end{enumerate}
\end{theorem}

\subsection{Combinatorial Characterization of Homeostasis}
\label{SS:combinatorial}

In order to classify these homeostasis subnetworks, we recall some  combinatorial properties on input-output networks.

\begin{definition} \normalfont
Let $\GG$ be a core input-output network.
\begin{enumerate}[(a)]
\item A directed path connecting node $\rho$ to node $\tau$ is called a {\em simple path} if it visits each node on the path at most once. Further, an {\em $\iota o$-simple path} is a simple path from the input node $\iota$ to the output node $o$.

\item A node in $\GG$ is {\em simple} if the node lies on an $\iota o$-simple path, and {\em appendage} if the node is not simple. Further, a simple node is called a {\em super-simple} node if it lies on every $\iota o$-simple path in $\GG$.

\item A simple path from node $\rho$ to node $\tau$ is an {\em appendage path} if every node on this path, except perhaps for $\rho$ and $\tau$, is an appendage node. \END
\end{enumerate}
\end{definition}

\begin{definition} \normalfont \label{def: appnet}
Let $\GG$ be a core input-output network.
\begin{enumerate}[(a)]
\item The {\em appendage subnetwork} $\AA_\GG$ of $\GG$ is the subnetwork consisting of all appendage nodes and all arrows in $\GG$ connecting appendage nodes.

\item The {\em complementary subnetwork} of an $\iota o$-simple path $S$ is the subnetwork $\CC_S$ consisting of all nodes not on $S$ and all arrows in $\GG$ connecting those nodes.

\item Nodes $\rho_i, \rho_j$ in $\AA_\GG$ are {\em path equivalent} if there exists paths in $\AA_\GG$ from $\rho_i$ to
$\rho_j$ and from $\rho_j$ to $\rho_i$. An {\em appendage path component} (or an appendage strongly connected component) is a path equivalence class in $\AA_\GG$.

\item Let $\AA\subset\AA_\GG$ be an appendage path component.
We say that $\AA$ satisfies the {\em no cycle condition} if for every $\iota o$-simple path $S$, nodes in $\AA$ do not form cycles with $\CC_S \setminus \AA$. 
\END
\end{enumerate}
\end{definition}

Nodes in the appendage subnetwork $\AA_\GG$ can be written uniquely as the disjoint union
\begin{equation} \label{e:AiBi io}
\AA_\GG = 
(\AA_1 \; \dot{\cup} \; \cdots \; \dot{\cup} \; \AA_s)\;  \; \dot{\cup} \; \; (\BB_1 \; \dot{\cup} \; \cdots \; \dot{\cup} \; \BB_t)
\end{equation}
where each $\AA_i$ is an appendage path component that satisfies the no cycle condition
and each $\BB_i$ is an appendage path component that violates the no cycle condition.
Moreover, each $\AA_i$ (resp. $\BB_i$) can be viewed as a subnetwork of $\AA_{\GG}$ by including the arrows in $\AA_{\GG}$ that connect nodes in $\AA_i$ (resp. $\BB_i$).
We call a component $\AA_i$ a {\em no cycle appendage path component}, and a component $\BB_i$ a {\em cycle appendage path component}. 

\begin{definition} \normalfont \label{def: supsimpnet io}
Let $\rho_1, \rho_2$ be adjacent super-simple nodes of $\GG$. 
\begin{enumerate}[(a)]
\item We say a simple node $\rho$ is {\em between} $\rho_1$ and $\rho_2$ if there exists an $\iota o$-simple path that includes $\rho_1$ to $\rho$ to $\rho_2$ in that order.

\item The {\em super-simple subnetwork} $\LL'(\rho_1,\rho_2)$ is the subnetwork of $\GG$ whose nodes are simple nodes between $\rho_1$ and $\rho_2$ and whose arrows are arrows of $\GG$ connecting nodes in $\LL(\rho_1,\rho_2)$.

\item The \textit{structural subnetwork} $\LL(\rho_1,\rho_2)$ is the subnetwork of $\GG$ generated by $\LL'(\rho_1,\rho_2)\cup \BB$, where $\BB$ consists of all cycle appendage path components that connect to $\LL'(\rho_1,\rho_2)$. \END
\end{enumerate}
\end{definition}

\begin{theorem}[\cite{WHAG21}]
\label{thm: appendage homeostasis subnetwork}
Let $\GG$ be a core input-output network.
\begin{enumerate}[{\rm (a)}]
\item Suppose $\AA_{\eta} \subset \AA_{\GG}$ is a no cycle appendage path component, then $\AA_{\eta}$ forms {\em an appendage homeostasis subnetwork} of $\GG$ and it is associated with an appendage homeostasis block.
\item Let $\rho_i, \rho_{i+1}$ be adjacent super-simple nodes in $\GG$. Then $\LL(\rho_i, \rho_{i+1})$ forms {\em a structural homeostasis subnetwork} of $\GG$ and it is associated with a structural homeostasis block.
\end{enumerate}
Conversely, the set of homeostasis subnetworks of $\GG$ is exactly the collection of subnetworks described in {\rm (a)} or {\rm (b)} above.
\end{theorem}

\subsection{Homeostasis Inducing and Homeostasis Patterns}
\label{SS:homeo_induce}

Suppose that infinitesimal homeostasis occurs in an input-output network, that is, $o'(\mathcal{I}_0)=0$ for some input value $\mathcal{I}_0$. 
In this case, we say that node $o$ is \emph{homeostatic at $\mathcal{I}_0$}.

\begin{definition} \normalfont
A \emph{homeostasis pattern} is the collection of all nodes that, in addition to the output node, are simultaneously forced to be homeostatic at $\mathcal{I}_0$. \END
\end{definition}

Since the output node $o$ is homeostatic when $\det(H) = 0$ at $(X_0,\II_0)$, thus at least one irreducible factor of $\det(H)$ vanishes.
We call a homeostasis subnetwork 
$\KK_\eta$ \emph{homeostasis inducing} if $h_{\KK_\eta} \equiv \det(B_\eta) = 0$ at $(X_0,\II_0)$. 
We let $\KK \Rightarrow \nu$ denote if node $\nu \in \GG$ is generically homeostatic whenever $\KK$ is homeostasis inducing (e.g. every homeostasis subnetwork in $\GG$ induces the output node $o$).
Given a subset of nodes $\NN$, if for every node $\nu \in \NN$, $\KK\Rightarrow \nu$, then we write $\KK \Rightarrow \NN$.

In \cite{DABGNRS23}, it proved that the induction relation is characterized by homeostasis subnetworks.
Furthermore, the induction applies in at least one direction for distinct homeostasis subnetworks, and no subnetwork induces itself.
Hence, there are four types of homeostasis inducing: structural / appendage homeostasis $\Rightarrow$ structural / appendage subnetworks.

Before stating the main theorem about the classification of homeostasis patterns, we need to introduce some terminology.

\begin{definition} \normalfont
\label{def: structural pattern network}
Let $\GG$ be an input-output network.
The {\em structural pattern network} $\PP_\sS$ of $\GG$ is the feedforward network whose nodes are the super-simple nodes $\rho_j$ and the \emph{backbone} nodes $\til{\LL}_j$ given by 
\[
\tLL_j \ \cup \{\rho_j, \rho_{j+1} \} = \LL (\rho_j, \rho_{j+1})
\]
where $\LL(\rho_i, \rho_{i+1})$ is a structural homeostasis subnetwork of $\GG$.
The arrows of $\PP_\sS$ are defined by the natural ordering of nodes in $\PP_\sS$ as
\begin{equation} \label{backbone_e}
\iota \to \tLL_1 \to \rho_{2} \to \tLL_2 \to \cdots \to \tLL_{q} \to o 
\end{equation}
If a structural homeostasis subnetwork only consists of two adjacent super-simple nodes (Haldane homeostasis type) and the arrow between them, then the corresponding backbone node is an empty set, which is still included in the structural pattern network $\PP_\sS$.
\END
\end{definition}

\begin{definition} \normalfont
\label{def: appendage component}
Each node $\til{\AA}$ in the {\em appendage pattern network} $\PP_\AA$ is a component in the condensation of appendage homeostasis subnetworks of $\GG$, and it is called an {\em appendage component}.
There is an arrow connecting appendage components $\til{\AA}_1$ and $\til{\AA}_2$ if there are nodes $\tau_1 \in \til{\AA}_1$ and $\tau_2 \in \til{\AA}_2$ such that $\tau_1 \to \tau_2 \in \GG$. \END
\end{definition}

\begin{definition} \normalfont
Given an appendage component $\til{\AA}$ whose nodes have arrows from or to simple nodes in $\GG$, then
\begin{enumerate}[(a)]
\item Suppose $\til{\VV}$ is the most downstream node in $\PP_\sS$, such that there is a simple node $\sigma \in \VV$ and an appendage path from a node $\tau \in \til{\AA}$ to the node $\sigma$. We choose uniquely an arrow from $\til{\AA}$ to the node in the structural pattern network $\til{\VV} \in \PP_\sS$.

\item Suppose $\til{\VV}$ is the most upstream node in $\PP_\sS$, such that there is a simple node $\sigma \in \VV$ and an appendage path from the node $\sigma$ to a node $\tau \in \til{\AA}$. We choose uniquely an arrow from the node in the structural pattern network $\til{\VV} \in \PP_\sS$ to $\til{\AA}$.
\END
\end{enumerate}
\end{definition}

\begin{definition} \normalfont
The \emph{homeostasis pattern network} $\PP$ of $\GG$ is the network whose nodes are the union of the nodes of the structural pattern network $\PP_\sS$ and the appendage pattern network $\PP_\AA$. The arrows of $\PP$ are the arrows of $\PP_\sS$, the arrows of $\PP_\AA$, and the arrows between them. \END
\end{definition}

Besides the super-simple nodes, there is a correspondence between the homeostasis subnetworks and their homeostasis pattern networks.
Each structural subnetwork corresponds to a backbone node. 
Each appendage subnetwork corresponds to an appendage component.
Let $\VV \subset \GG$ be a homeostasis subnetwork and let $\til{\VV} \in \PP$ be the corresponding node in the homeostasis pattern network. 
For any node $\nu \in \GG$ or a subset of nodes $\NN \subset \GG$, we let $\til{\VV} \Rightarrow \nu$ (resp. $\NN$) denote that $\VV$ {\em induces} $\nu$ (resp. $\NN$).

%Besides the super-simple nodes in $\GG$ and $\PP$, there is a correspondence between the homeostasis subnetworks in $\GG$ and the homeostasis pattern network $\PP$.
%Each structural subnetwork $\LL \subset \GG$  corresponds to a backbone node $\til{\LL} \in \PP_\sS$. 
%Each appendage subnetworks $\AA \subset \GG$ corresponds to a appendage component $\til{\AA} \in \PP_\AA$.
%
%Let $\VV \subset \GG$ be a homeostasis subnetwork and let $\til{\VV} \in \PP$ be the corresponding node in the homeostasis pattern network. We let  $\til{\VV} \Rightarrow \nu$ denote that $\VV$ {\em induces} $\nu$.

We start with some results of homeostasis patterns in \cite{DABGNRS23}.

\begin{theorem}[\cite{DABGNRS23}] \label{thm:pattern of homeo}
Suppose $\til{\AA} \in \PP_{\AA}$ is an appendage component, and $\tLL \in \PP_{\sS}$ is a backbone node in the homeostasis pattern network $\PP$, then
\begin{enumerate}[{\rm (a)}]
\item {\em (Structural Homeostasis $\Rightarrow$ Structural Subnetwork)} $\tLL$ induces every node of the structural pattern network $\PP_S$ downstream from $\tLL$, but no other nodes of $\PP_S$.
\item {\em (Structural Homeostasis $\Rightarrow$ Appendage Subnetwork)} Let $\til{\VV} \to \til{\AA} \in \PP$ with $\til{\VV}\in\PP_\sS$. 
$\tLL \Rightarrow \til{\AA}$ if and only if $\til{\VV}$ is strictly downstream from $\tLL$.
\item {\em (Appendage Homeostasis $\Rightarrow$ Structural Subnetwork)} Let $\til{\AA} \to \til{\VV} \in \PP$ with $\til{\VV} \in \PP_\sS$. $\til{\AA}$ induces every super-simple node downstream from $\til{\VV}$, but no other super-simple nodes. Further, $\til{\AA} \Rightarrow \tLL$ if and only if $\til{\VV}$ is strictly upstream from $\tLL$.
\item {\em (Appendage Homeostasis $\Rightarrow$ Appendage Subnetwork)} Let $\til{\AA}_1$ and $\til{\AA}_2$ be distinct appendage components. 
Let $\til{\AA}_1 \to \til{\VV}_1, \til{\VV}_2 \to \til{\AA}_2 \in \PP$ with $\til{\VV}_1, \til{\VV}_2 \in \PP_\sS$.
$\til{\AA}_1 \Rightarrow \til{\AA}_2$ if and only if 
$\til{\AA}_1$ is upstream from $\til{\AA}_2$ and every path from $\til{\AA}_1$ to $\til{\AA}_2$ contains a super-simple node which is downstream from $\til{\VV}_1$ and upstream from $\til{\VV}_2$.
\end{enumerate}
\end{theorem}

\section{Infinitesimal Homeostasis in PRN and GRN}
\label{S:inf homeo in RPN}

We now apply the infinitesimal homeostasis results of Wang \etal~\cite{WHAG21} to the PRN $\RR$.
We show that these results can be determined directly from graph theory on the GRN $\GG$ itself.

We start with the observation that, even though a GRN has no associated ODE, it is formally an abstract input-output network.
Hence, we can apply all the combinatorial constructions explained in subsection \ref{SS:combinatorial} to a GRN $\GG$ and its associated PRN $\RR$ (possible self-couplings of GRNs have no effect when carrying out these procedures).
The main goal of this section is to show how these combinatorial constructions relate to each other.

\subsection{Simple Paths in the PRN}
\label{SS:simple path in PRN}

We start by describing how the simple paths in the GRN relate to the simple paths in the PRN.
In turn, this allows us describe how to relate simple nodes, super simple nodes and appendage nodes among the two networks

\begin{lemma}  \label{simple path correspondence}
Let $\GG$ be an input-output GRN and $\RR$ be the associated input-output PRN.
\begin{enumerate}[{\rm (a)}]
\item \label{3.1a}  A path in $\GG$
\[
\sigma_1  \to \cdots \to \sigma_m
\]
is a simple if and only if the corresponding path
\[
\sigma_1^R \to \sigma_1^P \to \cdots \to \sigma_m^R \to \sigma_m^P
\] 
is a simple path in $\RR$
\item \label{3.1b} A path in $\GG$
\[
\iota \to \sigma_1  \to \cdots \to \sigma_m \to o
\]
is an $\iota o$-simple path if and only if the corresponding path
\[
\iota^R \to \iota^P \to \sigma_1^R \to \sigma_1^P \to \cdots \to \sigma_m^R \to \sigma_m^P \to o^R \to o^P
\]
is an $\iota^R o^P$-simple path in $\RR$.
\end{enumerate}
\end{lemma}

\begin{proof}
(a) Recall that in an input-output network, an $\iota o$-simple path is a simple path from the input node to the output node.  
Suppose there is an $\iota o$-simple path in $\GG$ as follows
\begin{equation} \label{e:simple_GRN}
\iota \to \sigma_1 \to \cdots \to \sigma_m \to o
\end{equation}
Note that self-coupling is never an arrow in a $\iota o$-simple path in $\GG$.  Thus, 
the $\iota o$-simple path \eqref{e:simple_GRN} lifts uniquely to the following $\iota^R o^P$-simple path in $\RR$:
\begin{equation} \label{e:simple_double}
\iota^R \to \iota^P \to \sigma_1^R \to \sigma_1^P \to \cdots \to \sigma_m^R \to \sigma_m^P \to o^R \to o^P
\end{equation}
On the other hand, since every $\iota^R o^P$-simple path consists of $2$-node blocks $j^R \to j^P$ for several GRN-nodes $j$, an $\iota^R o^P$-simple path in the PRN is always a lift of an $\iota o$-simple path in the GRN.

\smallskip
\noindent
(b) Similar as in item (a), it is clear that every simple path in $\GG$ lifts uniquely to a simple path in $\RR$.
Since every simple path also consists of $2$-node blocks $j^R \to j^P$ for several GRN-nodes $j$, every simple path in $\RR$ is always a lift of a simple path in $\GG$.
\end{proof}

\begin{remark} \em
No self-coupling arrow lies on a simple path in the GRN.
Similarly, the arrow  $\sigma^P\to\sigma^R$ never lies on a simple path in the PRN, otherwise such a path would contain $\sigma^R\to\sigma^P\to\sigma^R$ and hence would not be a simple path. 
\END
\end{remark} 

\begin{lemma} 
\label{node correspondence}
Let $\GG$ be an input-output GRN and $\RR$ be the associated input-output PRN.
\begin{enumerate}[{\rm (a)}]
\item Node $\tau$ is a simple node in $\GG$ if and only if the nodes $\tau^R, \tau^P$ are simple nodes in $\RR$.
\item Node $\tau$ is a super-simple node in $\GG$ if and only if the nodes $\tau^R, \tau^P$ are super-simple nodes in $\RR$.
\item Node $\tau$ is an appendage node in $\GG$ if and only if the nodes $\tau^R, \tau^P$ are appendage nodes in $\RR$.
\end{enumerate}
\end{lemma}

\begin{proof}
(a) Suppose that the node $\tau$ is simple in $\GG$, then there is an $\iota o$-simple path in $\GG$ containing $\tau$, such that
\begin{equation} \label{simple path in g}
\iota \to \sigma_1  \to \cdots \to \tau \to \cdots \to o
\end{equation}
By Lemma \ref{simple path correspondence}~\eqref{3.1b}, there exists an $\iota^R o^P$-simple path lifted from \eqref{simple path in g} in the PRN $\RR$ as
\begin{equation} \label{simple path in gd}
\iota^R \to \iota^P \to \sigma_1^R \to \sigma_1^P \to \cdots \to \tau^R \to \tau^P \to \cdots \to o^R \to o^P
\end{equation}
Thus, we obtain that both nodes $\tau^R, \tau^P$ are simple nodes in $\RR$.
On the other hand, assume both nodes $\tau^R, \tau^P$ are simple nodes in $\RR$, thus there exists an $\iota^R o^P$-simple path 
containing $\tau^R, \tau^P$ in $\RR$.
Using Lemma \ref{simple path correspondence}, this $\iota^R o^P$-simple path is a lift of an $\iota o$-simple path in the GRN containing $\tau$. 
This implies the node $\tau$ is simple in $\GG$.

\smallskip
\noindent
(b) The node $\tau$ is super-simple in $\GG$ if and only if every $\iota o$-simple path in $\GG$ contains $\tau$.
From Lemma \ref{simple path correspondence} (b), every $\iota^R o^P$-simple paths in $\RR$ is a lift of an $\iota o$-simple path in $\GG$. 
Thus, $\tau$ is super-simple in $\GG$ is equivalent to every $\iota^R o^P$-simple paths must contain both $\tau^R$ and $\tau^P$, which represents both nodes $\tau^R, \tau^P$ are super-simple nodes in $\RR$.

\smallskip
\noindent
(c) Since we know all nodes are either appendage or simple in $\GG$ and $\RR$.
Suppose $\tau$ is an appendage node. From part (a) we see that neither $\tau^R, \tau^P$ are simple nodes in $\RR$.
\end{proof}

\begin{lemma} \label{super simple order}
The set $\{\iota = \rho_1, \rho_2, \ldots, \rho_{q}, \rho_{q+1} = o\}$ of super-simple nodes in $\GG$
is well ordered by the order of their appearance on any $\iota o$-simple path.
We denote the ordering of the super-simple nodes by  
\begin{equation} \label{order of super simple in g}
\rho_1 \prec \cdots \prec \rho_{q} \prec \rho_{q+1}
\end{equation}
Then the ordering of the super-simple nodes $\{\rho^R_1, \rho^P_1, \ldots, \rho^R_{q+1}, \rho^P_{q+1}\}$ in $\RR$ is
\begin{equation} \label{order of super simple in gd}
    \rho^R_1 \prec \rho^P_1  \prec \cdots \prec \rho^R_{q+1} \prec \rho^P_{q+1}
\end{equation}
Moreover, $\LL (\rho_i^R, \rho_{i}^P) = \{ \rho_i^R \to \rho_{i}^P \}$ for $i = 1, 2, \ldots, q+1$.
\end{lemma}

\begin{proof}
From Lemma \ref{node correspondence}, we get $\rho^R_1, \rho^P_1, \ldots, \rho^R_{q+1}, \rho^P_{q+1}$ are the super-simple nodes in $\RR$. Since the order of super-simple nodes in $\GG$ is given in \eqref{order of super simple in g}, all $\iota o$-simple paths follow
\begin{equation}
\rho_1 \pathto \rho_2 \pathto \cdots \pathto \rho_{q} \pathto \rho_{q+1}
\end{equation}
where $\rho_j \pathto \rho_{j+1}$ indicates a simple path from $\rho_j$ to $\rho_{j+1}$.
Note that every $\iota^R o^P$-simple path in $\RR$ is always a lift of an $\iota o$-simple path in $\GG$. 
Each $\iota^R o^P$-simple path must satisfy
\begin{equation}
\rho^R_1 \pathto \rho^P_1 \pathto \cdots \pathto \rho^R_{q+1} \pathto \rho^P_{q+1}
\end{equation}
and we prove \eqref{order of super simple in gd}.
For any $i = 1, 2, \ldots, q+1$, the mRNA node $\rho_i^R$ only goes to the protein node $\rho_i^P$, and it is also the only node which has an arrow to $\rho_i^P$ in $\RR$. 
Therefore, as adjacent super-simple nodes $\rho_i^R, \rho_i^P$, we get that $\rho_i^R \to \rho_i^P$ is the only simple path, so $\LL (\rho_i^R, \rho_{i}^P) = \{ \rho_i^R \to \rho_{i}^P \}$.
\end{proof}

\begin{lemma} \label{appendage path in PRN}
Let $\sigma_1, \sigma_2$ be two distinct appendage nodes in $\GG$. Suppose there is an appendage path between $\sigma_1$ and $\sigma_2$.
Then $\sigma^P_1, \sigma^R_2$ are appendage nodes, and there exists an appendage path connecting $\sigma^P_1$ and $\sigma^R_2$ in $\RR$.
\end{lemma}

\begin{proof}
Since $\sigma_1, \sigma_2$ are appendage nodes in $\GG$, we obtain that $\sigma^P_1, \sigma^R_2$ are appendage nodes in $\RR$ from Lemma \ref{node correspondence}.
Assume the appendage path connecting $\sigma_1$ and $\sigma_2$ is
\begin{equation} \notag
\sigma_1 \to \tau_1  \to \cdots \to \tau_m \to \sigma_2
\end{equation}
where $\{\tau_i\}^{m}_{i=1}$ are appendage nodes in $\GG$.
Then there exists a corresponding path in $\RR$, such that
\begin{equation} \notag
\sigma_1^P \to \tau_1^R \to \tau_1^P \to \cdots \to \tau_m^R \to \tau_m^P \to \sigma_2^R
\end{equation}
Again from Lemma \ref{node correspondence}, $\{\tau^R_i, \tau^P_i\}^{m}_{i=1}$ are all appendage nodes in $\RR$. Therefore, we find the appendage path connecting two nodes $\sigma^P_1$ and $\sigma^R_2$.
\end{proof}

\subsection{Homeostasis Subnetworks in GRN}
\label{SS:homeo subnetworks in RPN}

In this section we prove the first main result of the paper, which gives the relation between the homeostasis subnetworks of the GRN and the associated PRN.

\begin{theorem} \label{thm: homeo subnetwork in PRN}
Let $\GG$ be an input-output GRN with associated input-output PRN $\RR$.
\begin{enumerate}[{\rm (a)}]
\item Every $\GG$-structural subnetwork $\LL(\rho_i,\rho_{i+1})$, where $\rho_i,\rho_{i+1}$ are consecutive super-simple nodes, corresponds to a $\RR$-structural subnetwork  $\LL(\rho^P_i,\rho^R_{i+1})$.
\item  Every super-simple node $\rho \in \GG$ corresponds to a $\RR$-Haldane subnetwork $\LL(\rho^R,\rho^P) = \{\rho^R\to\rho^P\}$, or $\LL(\rho^R,\rho^P) = \{\rho^R\rightleftarrows\rho^P\}$ if $\rho$ has a self-coupling.
\item  Every non-single appendage node $\GG$-subnetwork $\AA$ corresponds to a $\RR$-appendage subnetwork $\AA^{\RR}$.
In particular, a non-single appendage node $\GG$-subnetwork $\{\tau\selfarrow\}$ corresponds to a $\RR$-appendage subnetwork $\{\tau^R \rightleftarrows \tau^P\}$.
\item Every single appendage node $\GG$-subnetwork $\AA=\{\tau\}$ corresponds to two $\RR$-appendage subnetworks $\AA^R = \{\tau^R\}$ ($R$-null-degradation), $\AA^P = \{\tau^P\}$ ($P$-null-degradation).
\end{enumerate}
\end{theorem}

We start with the characterization of homeostasis subnetworks in GRN and PRN.

\begin{lemma} \label{subnetwork correspondence}
Let $\GG$ be an input-output GRN and $\RR$ be the associated input-output PRN.
Suppose that $\LL(\rho_i, \rho_{i+1})$ is a structural subnetwork and $\AA_j$ is an appendage subnetwork of $\GG$. 
\begin{enumerate}[{\rm (a)}]
\item The non-super-simple simple node $\tau \in \LL(\rho_i, \rho_{i+1})$ in $\GG$ if and only if both non-super-simple simple nodes $\{\tau^R, \tau^P\} \subset  \LL(\rho_i^P, \rho_{i+1}^R)$ in $\RR$.
\item The appendage node $\tau \in \LL (\rho_i, \rho_{i+1})$ in $\GG$ if and only if both appendage nodes $\{\tau^R, \tau^P\} \subset \LL (\rho_i^P, \rho_{i+1}^R)$ in $\RR$.
\item The node $\tau$ is a single appendage node in $\GG$ if and only if both nodes $\tau^R$ and $\tau^P$ form two appendage subnetworks in $\RR$.
\item The node $\tau \subset \AA_j$ or $\{\tau\} = \AA_j$ with self-coupling in $\GG$ if and only if both nodes $\{\tau^R, \tau^P\} \subset \AA_j^{\RR}$ or $\{\tau^R, \tau^P\} = \AA_j^{\RR}$ in $\RR$.
\end{enumerate}
\end{lemma}

\begin{proof}
Lemma \ref{super simple order} ensures that $\rho_i^R, \rho_i^P, \rho_{i+1}^R, \rho_{i+1}^P$ are four adjacent super-simple nodes in $\RR$.

\smallskip
\noindent
(a) Suppose $\tau \in \LL(\rho_i, \rho_{i+1})$ is a non-super-simple simple node in $\GG$, then there exists a $\rho_i \rho_{i+1}$-simple path containing $\tau$. 
Using Lemma \ref{node correspondence}, we obtain that $\tau^R, \tau^P$ are non-super-simple simple nodes, which are contained in the corresponding $\rho_i^R \rho_{i+1}^P$-simple path lifted from GRN in $\RR$. Thus, $\{\tau^R, \tau^P\} \subset \LL(\rho_i^P, \rho_{i+1}^R)$ in $\RR$.
On the other hand, suppose both non-super-simple simple nodes $\{\tau^R, \tau^P\} \subset \LL(\rho_i^P, \rho_{i+1}^R)$ in $\RR$. 
There must exist a $\rho_i^R \rho_{i+1}^P$-simple path containing $\tau^R, \tau^P$ in $\RR$, which is a lift of an $\rho_i \rho_{i+1}$-simple path in the GRN containing the non-super-simple simple node $\tau$.

\smallskip
\noindent
(b) Assume $\tau \in \LL (\rho_i, \rho_{i+1})$ is an appendage node in $\GG$, then there exists a cycle consisting of at least one non-super-simple simple and appendage nodes including $\tau$, that is,
\begin{equation} \label{simple in i cycle in g} 
\tau \to \sigma_1  \to \cdots \to \sigma_m \to \tau
\end{equation}
where $\sigma_j \in \LL (\rho_i, \rho_{i+1})$ for $j = 1, \ldots, m$ and at least one node in cycle is non-super-simple simple.
Now, we can assume (w.l.o.g.) that $\sigma_1 \in \LL (\rho_i, \rho_{i+1})$, and obtain the following cycle in $\RR$:
\begin{equation} \label{simple in i cycle in gd} 
\tau^R \to \tau^P \to \sigma^R_1  \to \cdots \to \sigma^P_m \to \tau^R
\end{equation}
It is clear that $\sigma^R_1 \neq \rho^R_{i+1}$ and $\sigma^P_1 \neq \rho^P_i$.
From part $(a)$, we derive $\{\sigma_1^R, \sigma_1^P\} \subset \LL(\rho_i^P, \rho_{i+1}^R)$ in $\RR$. 
This works for all other non-super-simple simple nodes in \eqref{simple in i cycle in g} as well. 
Hence, \eqref{simple in i cycle in gd} is a cycle consisting of at least two non-super-simple simple and appendage nodes including $\tau^R, \tau^P$.  
Therefore, $\{\tau^R, \tau^P\} \subset  \LL(\rho_i^P, \rho_{i+1}^R)$ in $\RR$.
Now, suppose that both appendage nodes $\{\tau^R, \tau^P\} \subset  \LL (\rho_i^P, \rho_{i+1}^R)$ in $\RR$. 
There must exist a cycle in $\LL (\rho_i^P, \rho_{i+1}^R)$ containing $\tau^R, \tau^P$ but $\rho_i^P, \rho_{i+1}^R$.
By Lemma \ref{node correspondence} and item $(a)$, this must be a lift of a cycle in $\LL (\rho_i, \rho_{i+1})$ containing $\tau$ but $\rho_i, \rho_{i+1}$.

\smallskip
\noindent
(c) Suppose that $\tau$ is a single appendage node in $\GG$, then 
$\AA_j=\{\tau\}$ and there is no cycle containing $\tau$ in $\AA_\GG$. 
By item (b), since the node $\tau$ is not self-coupling, there is no cycle containing $\tau^R$ and $\tau^P$ in $\AA_{\RR}$.
Therefore $\tau^R$ and $\tau^P$ form two separate appendage subnetworks in $\RR$.
Next, suppose that both appendage nodes $\{\tau^R\}$ and $\{\tau^P\}$ form two appendage subnetworks in $\RR$. 
Thus, there is no cycle containing $\tau^R$ and $\tau^P$ in $\AA_{\RR}$.
From item (b), it follows that node $\tau$ isn't self-coupling, and all cycles in $\AA_\GG$ exclude $\tau$.

\smallskip
\noindent
(d) We first consider an appendage subnetwork $\AA_j=\{\tau\}$ consisting of a single appendage node with self-coupling. 
It follows from item (c) that there is no cycle containing $\tau^R$ and $\tau^P$ in $\AA_{\RR}$, except the cycle $\tau^R  \to \tau^P \to \tau^R$ from self-coupling on $\tau$. 
Then $\{\tau^R, \tau^P\} = \AA_j^{\RR}$ form an appendage subnetwork in $\RR$.
Next, suppose that the appendage node $\tau \subset \AA_j$ in $\GG$, there exists a cycle consisting of appendage nodes including $\tau$ in $\AA_j \subseteq \AA_{\GG}$ as follows
\begin{equation} \label{appendage in i cycle in g} 
\tau \to \tau_1  \to \cdots \to \tau_n \to \tau
\end{equation}
where $\tau_i \in \AA_j$ for $i = 1, \ldots, n$. 
From items (b) and (c), this lifts uniquely to a corresponding cycle in $\AA_{\RR}$:
\begin{equation} \label{appendage in i cycle in gd} 
\tau^R \to \tau^P \to \tau^R_1  \to \cdots \to \tau^P_n \to \tau^R
\end{equation}
thus $\{\tau^R, \tau^P\} \subset \AA^{\RR}_j$ in $\RR$.

\noindent
The converse direction of item (d)  follows from items (b) and (c).
\end{proof}

\begin{proof}[Proof of Theorem \ref{thm: homeo subnetwork in PRN}]
Let $\GG$ be an input-output GRN with $n+2$ nodes and
with $q+1$ super-simple nodes $\iota = \rho_1 \prec \cdots \prec \rho_{q} \prec \rho_{q+1} = o$.
Then $\GG$ has $q$ structural homeostasis subnetworks (self-couplings can be kept during these procedures)
\begin{equation} 
\LL(\rho_1,\rho_2), \; \LL(\rho_2,\rho_3), \; \ldots, \; \LL(\rho_{q-1},\rho_q), \; \LL(\rho_q,,\rho_{q+1})
\end{equation}
where $\LL(\rho_{i},\rho_{i+1}) = \LL' (\rho_i,\rho_{i+1}) \; \cup \; \BB_{i, i+1}$, with 
$\BB_{i, i+1}$ consisting of all appendage path components in $\AA_\GG$ that violate the no cycle condition with respect to simple nodes in $\LL'(\rho_i,\rho_{i+1})$.
Moreover, $\GG$ has $r+s$ appendage homeostasis subnetworks
\begin{equation} 
\AA'_1, \ldots, \AA'_r, \; 
\AA_1, \ldots, \AA_s
\end{equation}
where $\AA'_i = \{ \tau_i \}$ and $\tau_i$ is a single appendage node for $1 \leq i \leq r$.

First, we consider appendage homeostasis subnetworks.
In $\GG$, the appendage homeostasis subnetworks $\AA'_i$  consists of a single appendage node $\tau_i$ for $1 \leq i \leq r$.
From Lemma \ref{subnetwork correspondence}(c), nodes $\tau^R_i$ and $\tau^P_i$ form two appendage homeostasis subnetworks in $\RR$, which correspond to two irreducible blocks of the form
\[
\text{R-null-degradation: } \big[ f_{\tau^R_i, \tau^R_i} \big]
\ \text{ and } \
\text{P-null-degradation: } \big[ f_{\tau^P_i, \tau^P_i} \big]
\]

For the rest of appendage homeostasis subnetworks in $\GG$, we assume that
\[
\AA_j := \{\sigma_{j_1}, \ldots, \sigma_{j_k} \}, \ \text{for } 1 \leq j \leq s
\]
Applying Lemma \ref{subnetwork correspondence}(d), we obtain the corresponding appendage homeostasis subnetworks $\AA_j^{\RR}$ in $\RR$ as follows (including appendage nodes with self-coupling)
\[
\AA_j^{\RR} = \{\sigma^R_{j_1}, \sigma^P_{j_1}, \ldots, \sigma^R_{j_k}, \sigma^P_{j_k} \}
\]

Now we deal with structural homeostasis subnetworks.
Given super-simple nodes $\iota = \rho_1 \prec \cdots \prec \rho_{q} \prec \rho_{q+1} = o$ in $\GG$, from Lemma \ref{super simple order} we have that $\rho^R_1, \rho^P_1, \ldots, \rho^R_{q+1}, \rho^P_{q+1}$ are all super-simple nodes in $\RR$, and $\LL(\rho_i^R, \rho_{i}^P) = \langle\rho_i^R, \rho_{i}^P \rangle$ for $1 \leq i \leq q+1$. 
Clearly, $\LL(\rho_i^R, \rho_{i}^P)$ has the input node $\rho_i^R$ and the output node $\rho_i^P$, thus it corresponds to the following irreducible components in $B$:
\[
\text{$\RR$-Haldane: } \big[ f_{\rho_i^P, \rho_i^R} \big]
\]

For other structural homeostasis subnetworks in $\GG$, we assume that
\[
\LL (\rho_j, \rho_{j+1}) := \{ \rho_j, \rho_{j+1}, \sigma_{j_1}, \ldots, \sigma_{j_k} \}, \ \text{for } 1 \leq j \leq q
\]
Using Lemma \ref{subnetwork correspondence}, we obtain the corresponding structural homeostasis subnetworks in $\RR$
\[
\LL(\rho^P_{i}, \rho^R_{i+1}) = \{\rho^P_{i}, \rho^R_{i+1}, \sigma^R_{i_1}, \sigma^P_{i_1}, \ldots, \sigma^R_{i_l}, \sigma^P_{i_l} \}
\]
In summary, from the doubling of nodes, it follows that the associated PRN $\RR$ has $2q$ structural homeostasis subnetworks 
\begin{equation}
\LL (\rho_1^R,\rho_1^P), \; \LL(\rho_1^P,\rho_2^R), \; \ldots, \; \LL(\rho_{q}^P,\rho_{q+1}^R), \; \LL (\rho_{q+1}^R,\rho_{q+1}^P)
\end{equation}
where $\LL(\rho_i^R,\rho_i^P) = \{\rho_i^R \to \rho_i^P \}$, or  
$\LL(\rho_i^R,\rho_i^P) = \{\rho_i^R \rightleftarrows \rho_i^P \}$ if the super-simple node $\rho$ has a self-coupling, for $1 \leq i \leq q+1$ are $\RR$-Haldane subnetworks. 
Moreover, $\RR$ has $2r+s$ appendage homeostasis subnetworks
\begin{equation} 
\AA^R_1, \AA^P_1, \ldots,\AA^R_r, \AA^P_r, \; 
\AA^{\RR}_1, \ldots, \AA^{\RR}_s
\end{equation}
where
$\AA^R_i = \{\tau^R_i\}$, $\AA^P_i = \{\tau^P_i\}$ for $1 \leq i \leq r$.
This gives a complete correspondence between the homeostasis subnetworks of a GRN and its associated PRN.
\end{proof}

\subsection{Enumerating Homeostasis Subnetworks in GRN and PRN} 
\label{SS:algorithm}

The following algorithm is used to find the different homeostasis subnetworks of a GRN $\GG$, with an input node $\iota$ and an output node $o$, and the homeostasis subnetworks of the associated PRN $\RR$.

\paragraph{Step 0:} Reduce the input-output network to a core GRN network $\GG$ and then let $\RR$ be the RPN of the core GRN.  This process is the same as first forming the RPN of the GRN and then reducing to the core network.

\paragraph{Step 1:} Identify the $\iota o$-simple paths in the core GRN network $\GG$ and the simple nodes $\sigma$, the super-simple nodes $\rho$, and the appendage nodes $\tau$ of $\GG$.
The self-couplings of simple $\GG$-nodes can be removed.
It follows from Lemma \ref{simple path correspondence} that $\{\sigma^R,\sigma^P\}$ are simple nodes, $\{\rho^R,\rho^P\}$ are super-simple nodes, and $\{\tau^R,\tau^P\}$ are appendage nodes of $\RR$.

\paragraph{Step 2:} Determine the appendage homeostasis subnetworks of $\GG$.
Specifically, the appendage subnetwork $\AA_\GG$ of $\GG$ can be written uniquely as the disjoint union
\begin{equation} \label{e:AiBi}
\AA_\GG = 
(\AA'_1 \; \dot{\cup} \; \cdots \; \dot{\cup} \; \AA'_r)\;  \; \dot{\cup} \; \; (\AA_1 \; \dot{\cup} \; \cdots \; \dot{\cup} \; \AA_s)\;  \; \dot{\cup} \; \; (\BB_1 \; \dot{\cup} \; \cdots \; \dot{\cup} \; \BB_t)
\end{equation}
where each $\AA'_i$ consists of a single appendage node, each $\AA_i$ is a no cycle appendage path component,
and each $\BB_i$ is an appendage path component that violates the no cycle condition.
The appendage homeostasis subnetworks of $\GG$ are
\begin{equation} \label{appendage homeostasis subnetworks in g}
\AA'_1, \ldots, \AA'_r, \AA_1, \ldots, \AA_s
\end{equation}
By Lemma \ref{subnetwork correspondence}, the corresponding appendage subnetwork $\AA_\RR$ of $\RR$ can be written as
\begin{equation} \label{e:AiBid}
\AA_\RR = 
(\AA_{1}^R \; \dot{\cup} \; \AA_1^P \; \dot{\cup} \; \cdots \; \dot{\cup} \; \AA_r^R \; \dot{\cup} \; \AA_r^P)\;  \; \dot{\cup} \; \;  (\AA_1^{\RR} \; \dot{\cup} \; \cdots \; \dot{\cup} \; \AA_s^{\RR})\;  \; \dot{\cup} \; \;(\BB_1^{\RR} \; \dot{\cup} \; \cdots \; \dot{\cup} \; \BB_t^{\RR})
\end{equation}
where 
\begin{equation}
\begin{split}
\{\tau^R,\tau^P\}\subset \AA_j^{\RR} \ (\text{or }\BB_j^{\RR}) & \ \text{ if } \
\tau\in\AA_j \ (\text{or } \BB_j)
\\ \AA_j^R = \{\tau^R\}, \
\AA_j^P = \{\tau^P\} & \ \text{ if } \ \AA'_j = \{\tau\}
\end{split}
\end{equation}
The appendage homeostasis subnetworks of $\RR$ are
\begin{equation} \label{appendage homeostasis subnetworks in gd}
\AA_1^R, \AA_1^P, \ldots, \AA_r^R, \AA_r^P, \AA_1^{\RR}, \ldots, \AA_s^{\RR}
\end{equation}
In particular, if $\tau$ is an appendage node with self-coupling which gives an appendage subnetwork $\{\tau\selfarrow\}$ of $\GG$ then the corresponding appendage network of $\RR$ is $\{\tau^R \rightleftarrows	\tau^P\}$.

\paragraph{Step 3:} Determine the structural homeostasis subnetworks of $\GG$.
Let
\begin{equation} \label{e:rho1toq}
 \iota = \rho_1 \prec \cdots \prec \rho_{q} \prec \rho_{q+1} = o
\end{equation}
be the ordered set of super-simple nodes of $\GG$.
The super-simple subnetworks of $\GG$ are
\begin{equation} \label{super-simple subnetworks in g}
\LL' (\rho_1,\rho_2),\; \LL' (\rho_2,\rho_3) \;,\ldots,\; \LL' (\rho_{q-1},\rho_q),\; \LL' (\rho_q,\rho_{q+1})
\end{equation}
Then, the structural subnetworks of $\GG$ are
\begin{equation} \label{e:SHB}
\LL(\rho_1,\rho_2),\; \LL(\rho_2,\rho_3) \;,\ldots,\; \LL(\rho_{q-1},\rho_q),\; \LL(\rho_q,,\rho_{q+1})
\end{equation}
where $\LL(\rho_{i},\rho_{i+1}) = \LL' (\rho_i,\rho_{i+1}) \cup \BB$, with $\BB$ consists of all appendage path components that violate the no cycle condition with simple nodes in $\LL'(\rho_i,\rho_{i+1})$.

By Lemma \ref{super simple order}, the ordered set of super-simple nodes of $\RR$ are 
\[
 \iota  = \rho_1^R \prec \rho_1^P \prec \cdots \prec \rho_{q+1}^R \prec \rho_{q+1}^P = o
\]
By Lemma \ref{subnetwork correspondence} the super-simple subnetworks of $\RR$ are
\begin{equation} \label{super-simple subnetworks in gd}
\LL' (\rho_1^R,\rho_1^P), \LL' (\rho_1^P,\rho_2^R), \;\ldots,\; \LL' (\rho_{q}^P,\rho_{q+1}^R),\; \LL' (\rho_{q+1}^R,\rho_{q+1}^P)
\end{equation}
where 
\begin{equation}
\LL'(\rho_i^R,\rho_i^P) = \langle\rho_i^R, \rho_i^P\rangle
\ \text{ and } \
\{\tau^R,\tau^P\} \subset \LL' (\rho_i^P,\rho_{i+1}^R) \ \text{ if } \
\tau \in \LL' (\rho_i,\rho_{i+1})
\end{equation}
Then, the structural subnetworks of $\RR$ are
\begin{equation} \label{structural homeostasis subnetworks in gd}
\LL (\rho_1^R,\rho_1^P), \LL(\rho_1^P,\rho_2^R),\;\ldots,\; \LL(\rho_{q}^P,\rho_{q+1}^R),\; \LL (\rho_{q+1}^R,\rho_{q+1}^P)
\end{equation}
where $\LL(\rho_i^R,\rho_i^P) = \{\rho_i^R \to \rho_i^P\}$, or
$\LL(\rho_i^R,\rho_i^P) = \{\rho_i^R \rightleftarrows	 \rho_i^P\}$,
and $\LL(\rho_{i}^P,\rho_{i+1}^R) = \LL' (\rho_i^R,\rho_i^P) \cup \BB^{\RR}$, with $\BB^{\RR}$ consisting of all appendage path components that violate the no cycle condition with respect to simple nodes in $\LL'(\rho_i^R,\rho_i^P)$.

\section{Homeostasis Patterns in PRN and GRN}
\label{S:pattern of homeo}

In this section, we obtain the homeostasis in the PRN and show how they relate to the homeostasis patterns in the GRN.

\subsection{Homeostasis Pattern Networks}
\label{SS: homeo inducing in the RPN}

We start by showing the relation between the {\em homeostasis pattern network} $\PP(\GG)$ of the GRN and the {\em homeostasis pattern network} $\PP(\RR)$ of the associated PRN.
We start by relating the nodes in both networks.

First, we establish the relation among the nodes in structural pattern network $\PP_\sS$ and the appendage components $\PP_\AA$ in both homeostasis pattern networks $\PP(\GG)$ and $\PP(\RR)$.

\begin{lemma}
\label{node in homeostasis pattern network}
Let $\GG$ be an input-output GRN and $\RR$ be the associated input-output PRN.
\begin{enumerate}[{\rm (a)}]
\item The nodes of the structural pattern network $\PP_\sS (\GG)$ are
\begin{equation}
\iota = \rho_1, \; \tLL_1,\; \rho_2, \; \tLL_2,\; \ldots,\;  \tLL_q,\; \rho_{q+1} = o
\end{equation}
where $\tLL_j \cup \{\rho_j, \rho_{j+1}\} = \LL (\rho_j, \rho_{j+1})$  for $1 \leq j \leq q$.

\noindent
The appendage components of the appendage pattern network $\PP_\AA (\GG)$ are
\begin{equation} 
\til{\AA}'_1, \ldots, \til{\AA}'_r, \; 
\til{\AA}_1, \ldots, \til{\AA}_s
\end{equation}
where $\til{\AA}'_i = \til{\AA}'_i$ for $1 \leq i \leq r$, and $\til{\AA}_j = \til{\AA}_j$ for $1 \leq j \leq s$.

\item The nodes of the structural pattern network $\PP_\sS (\RR)$ are
\begin{equation} 
\iota = \rho^R_1, \ \tLL (\rho_1^R,\rho_1^P), \ \rho^P_1, \ \tLL (\rho_1^P,\rho_2^R), \ \rho_2^R, \ \ldots, \ \tLL (\rho_{q+1}^R,\rho_{q+1}^P), \ \rho^P_{q+1} = o
\end{equation}
where the backbone nodes $\tLL$ are obtained by 
\begin{equation}
\begin{split}
& \tLL (\rho_i^R,\rho_i^P) = \emptyset, \ \text{for } i = 1, \ldots, q+1
\\& \tLL (\rho_j^P, \rho_{j+1}^R) \cup 
\{\rho_j^P, \rho^R_{j+1}\} = \LL (\rho_j^P, \rho_{j+1}^R), \ \text{for } j = 1, \ldots, q
\end{split}
\end{equation}
The appendage components of the appendage pattern network $\PP_\AA (\RR)$ are
\begin{equation} 
\til{\AA}^R_1, \ldots, \til{\AA}^P_r, \; 
\til{\AA}^{\RR}_1, \ldots, \til{\AA}^{\RR}_s
\end{equation}
where $\til{\AA}^R_i = \AA^R_i, \til{\AA}^P_i = \AA^P_i$ for $1 \leq i \leq r$, and $\til{\AA}^{\RR}_j = \AA^{\RR}_j$ for $1 \leq j \leq s$.
\end{enumerate}
\end{lemma}

\begin{proof}
(a) It follows straightly from the Definition 
\ref{def: structural pattern network} and \ref{def: appendage component}.

\smallskip
\noindent
(b) The super-simple nodes in $\RR$ are
\[
\rho^R_1, \rho^P_1, \ldots, \rho^R_{q+1}, \rho^P_{q+1}
\] 
and there are two types of structural homeostasis subnetworks
\begin{equation}
\begin{split}
& \LL (\rho_i^R,\rho_i^P) = \{\rho_i^R, \rho_i^P \}, \ \text{for } i = 1, \ldots, q+1
\\& \LL(\rho_j^P,\rho_{j+1}^R) \supseteq \{\rho_j^P, \rho^R_{j+1}\}, \ \text{for } j = 1, \ldots, q
\end{split}
\end{equation}
Note that every structural homeostasis subnetwork $\LL(\rho_i^R,\rho_i^P)$ only consists of two adjacent super-simple nodes, thus the corresponding backbone node $\tLL (\rho_i^R,\rho_i^P)$ is an empty set.
Next, from Definitions \ref{def: structural pattern network} and \ref{def: appendage component}, we obtain the rest of the backbone nodes of $\PP_\sS (\RR)$ and the appendage components of $\PP_{\AA} (\RR)$ in the PRN $\RR$.
\end{proof}

Next, we clarify the relation between the arrows in both homeostasis pattern networks.

\begin{lemma} \label{arrow of structural pattern network or appendage component}
Let $\GG$ be an input-output GRN and $\RR$ be the associated input-output PRN.
\begin{enumerate}[{\rm (a)}] 
\item The arrows of the structural pattern network $\PP_\sS (\RR)$ are
\begin{equation} 
\label{arrows of the structural pattern network in R}
\rho^R_1 \to \tLL (\rho_1^R,\rho_1^P) \to \rho^P_1 \to \tLL (\rho_1^P,\rho_2^R) \to \rho^R_2 \to
\cdots \to \tLL (\rho_{q+1}^R,\rho_{q+1}^P) \to \rho^P_{q+1}
\end{equation}

\item Consider two distinct appendage components $\til{\AA}_1, \til{\AA}_2$ in $\PP_{\AA} (\GG)$. 
Suppose $\til{\AA}_1 \to \til{\AA}_2 \in \PP_{\AA} (\GG)$, then the corresponding arrows in $\PP_{\AA} (\RR)$ are:
\begin{enumerate}[{\rm (i)}]
\item If $\til{\AA}_1 \in \{ \til{\AA}'_i \}^{r}_{i=1}$ and $\til{\AA}_2 \in \{ \til{\AA}_j \}^{s}_{j=1}$, then
\begin{equation}
\label{arrows of appendage component in R part1}
\til{\AA}^{R}_1 \to \til{\AA}^{P}_1 \to \til{\AA}^{\RR}_2
\end{equation}

\item If $\til{\AA}_2 \in \{ \til{\AA}'_i \}^{r}_{i=1}$ and $\til{\AA}_1 \in \{ \til{\AA}_j \}^{s}_{j=1}$, then
\begin{equation}
\label{arrows of appendage component in R part2}
\til{\AA}^{\RR}_1 \to \til{\AA}^{R}_2 \to \til{\AA}^{P}_2
\end{equation}

\item If $\til{\AA}_1, \til{\AA}_2 \in \{ \til{\AA}'_i \}^{r}_{i=1}$, then
\begin{equation}
\label{arrows of appendage component in R part3}
\til{\AA}^{R}_1 \to \til{\AA}^{P}_1 \to \til{\AA}^{R}_2 \to \til{\AA}^{P}_2
\end{equation}

\item If $\til{\AA}_1, \til{\AA}_2 \in \{ \til{\AA}_j \}^{s}_{j=1}$, then
\begin{equation}
\label{arrows of appendage component in R part4}
\til{\AA}^{\RR}_1 \to \til{\AA}^{\RR}_2
\end{equation}
\end{enumerate}
\end{enumerate}
\end{lemma}

\begin{proof}
(a) From Definition \ref{def: structural pattern network} the arrows of $\PP_\sS (\GG)$ satisfy
\begin{equation} 
\iota \to \tLL_1 \to \rho_{2} \to \tLL_2 \to \cdots \to \tLL_{q} \to o 
\end{equation}
where $\iota = \rho_1 \prec \rho_2 \prec \ldots \prec \rho_{q+1} = o$ are the super-simple nodes in $\GG$.
Applying Lemma \ref{super simple order}, we get the super-simple nodes in $\RR$ under the following order, 
\begin{equation} 
\rho^R_1 \prec \rho^P_1  \prec \cdots \prec \rho^R_{q+1} \prec \rho^P_{q+1}
\end{equation}
Again from Definition \ref{def: structural pattern network}, we get \eqref{arrows of the structural pattern network in R}.

\smallskip
\noindent
(b) For \eqref{arrows of appendage component in R part1}, it follows from $\til{\AA}_1 \in \{ \til{\AA}'_i \}^{r}_{i=1}$ and $\til{\AA}_2 \in \{ \til{\AA}_j \}^{s}_{j=1}$ that we can assume $\til{\AA}_1 = \{\tau_1 \}$ with $\tau_1$ a single appendage node in $\GG$. 
Then we get the corresponding appendage components in $\PP_{\AA}(\RR)$ as
\[
\til{\AA}^{R}_1 = \{\tau^R_1\}, \; \til{\AA}^{P}_1 = \{\tau^P_1\}, \; \til{\AA}^{\RR}_2
\]
From $\til{\AA}_1 \to \til{\AA}_2$ in $\PP_{\AA} (\GG)$, there exists a node $\sigma \in \til{\AA}_2$, such that $\tau_1 \to \sigma$.
Using Lemma \ref{subnetwork correspondence}, we have $\{ \sigma^R, \sigma^P \} \in \til{\AA}^{\RR}_2$.
Under the gene coupling in the PRN, we get arrows $\tau^{R}_1 \to \tau^{P}_1$ and $\tau^{P}_1 \to \sigma^R$, which implies that
\[
\til{\AA}^{R}_1 \to \til{\AA}^{P}_1 \to \til{\AA}^{\RR}_2
\]

\smallskip
\noindent
For \eqref{arrows of appendage component in R part2}, since $\til{\AA}_2 \in \{ \til{\AA}'_i \}^{r}_{i=1}$ and $\til{\AA}_1 \in \{ \til{\AA}_j \}^{s}_{j=1}$, we assume $\til{\AA}_2 = \{\tau_2 \}$ with $\tau_2$ is a single appendage node in $\GG$. Then we get the corresponding appendage components in $\PP_{\AA} (\RR)$ as
\[
\til{\AA}^{\RR}_1, \; \til{\AA}^{R}_2 = \{\tau^R_2\}, \; \til{\AA}^{P}_2 = \{\tau^P_2\}
\]
Again from $\til{\AA}_1 \to \til{\AA}_2$ in $\PP_{\AA} (\GG)$, there exists a node $\sigma \in \til{\AA}_1$, such that $\sigma \to \tau_2$.
Using Lemma \ref{subnetwork correspondence}, we have $\{ \sigma^R, \sigma^P \} \in \til{\AA}^{\RR}_1$.
Under the gene coupling in the PRN, we get arrows $\tau^{R}_2 \to \tau^{P}_2$ and $\sigma^R \to \sigma^P \to \tau^{R}_2$, and obtain 
\[
\til{\AA}^{\RR}_1 \to \til{\AA}^{R}_2 \to \til{\AA}^{P}_2
\]
We omit the proofs of \eqref{arrows of appendage component in R part3} and \eqref{arrows of appendage component in R part4}, since they follow directly from the above.
\end{proof}

\begin{lemma} \label{arrow between structural pattern network and appendage component 1}
Let $\GG$ be an input-output GRN and $\RR$ be the associated input-output PRN.
Consider $\til{\AA} \in \PP_{\AA} (\GG)$ and $\VV \in \PP_{\sS} (\GG)$. 
Suppose $\til{\AA} \to \VV \in \PP (\GG)$, then the corresponding arrows in $\PP (\RR)$ are
\begin{enumerate}[{\rm (a)}]
\item If $\til{\AA} \in \{ \til{\AA}'_i \}^{r}_{i=1}$ and $\VV = \rho_i $, then
\begin{equation}
\label{arrows from appendage to structural in R part1}
\til{\AA}^{R} \to \til{\AA}^{P} \to  \rho^R_i
\end{equation}

\item If $\til{\AA} \in \{ \til{\AA}_j \}^{s}_{j=1}$ and $\VV = \tLL_i$, then
\begin{equation}
\label{arrows from appendage to structural in R part2}
\til{\AA}^{\RR} \to \tLL (\rho_i^P, \rho_{i+1}^R)
\end{equation}

\item If $\til{\AA} \in \{ \til{\AA}_j \}^{s}_{j=1}$ and $\VV = \rho_i$, then
\begin{equation}
\label{arrows from appendage to structural in R part3}
\til{\AA}^{\RR} \to  \rho^R_i
\end{equation}

\item If $\til{\AA} \in \{ \til{\AA}'_i \}^{r}_{i=1}$ and $\VV = \tLL_i$, then
\begin{equation}
\label{arrows from appendage to structural in R part4}
\til{\AA}^{R} \to \til{\AA}^{P} \to \tLL (\rho_i^P, \rho_{i+1}^R)
\end{equation}
\end{enumerate}
\end{lemma}

\begin{proof}
(a) From $\til{\AA} \in \{ \til{\AA}'_i \}^{r}_{i=1}$ and $\VV = \{ \rho_i \}$, we may assume that $\til{\AA} = \{\tau_1 \}$ with $\tau_1$ is a single appendage node in $\GG$. Then we get the corresponding appendage components in $\PP_{\AA} (\RR)$ as
\[
\til{\AA}^{R} = \{\tau^R_1 \}, \; \til{\AA}^{P} = \{\tau^P_1\}
\]
and the nodes in $\PP_\sS (\RR)$ as
\[
\rho^R_i, \; \til{\LL} (\rho^R_i, \rho^P_i), \; \rho^P_i
\]
From the arrow $\tau_1 \to \rho_i$, $\rho_i$ is the most upstream simple node in $\GG$ that allows an appendage path from $\tau_1$.
Using appendage paths correspondence in Corollary \ref{appendage path in PRN}, we obtain that $\rho_i^R$ is the most upstream simple node in $\RR$ which allows an appendage path from $\til{\AA}^{P}$.
Together with the gene coupling in the PRN, we obtain arrows 
\[
\til{\AA}^{R} \to \til{\AA}^{P} \to \rho^R_i
\]

\smallskip
\noindent
(b) From $\til{\AA} \in \{ \til{\AA}_j \}^{s}_{j=1}$ and $\tLL = \tLL_i$, we get the corresponding appendage component $\til{\AA}^{\RR}$ in $\PP_{\AA} (\RR)$ 
and the backbone node $\tLL (\rho_i^P, \rho_{i+1}^R)$ in $\PP_\sS (\RR)$.
Moreover, from the arrow $\til{\AA} \to \tLL_i$ in $\PP (\GG)$, there exists two gene nodes $\tau \in \til{\AA}$ and $\sigma \in \tLL_i$, such that 
\begin{equation} \notag
\tau \to \sigma
\end{equation}
where $\sigma$ is the most upstream simple node which allows an appendage path from nodes in appendage component $\til{\AA}$. 
Using Lemma \ref{subnetwork correspondence} and the gene coupling in PRN, we have 
\begin{equation} \notag
\{ \tau^R, \tau^P \} \in \til{\AA}^{\RR}, \; \{ \sigma^R, \sigma^P \} \in \tLL (\rho_i^P, \rho_{i+1}^R), \; \text{and } \tau^P \to \sigma^R
\end{equation}
Again from Corollary \ref{appendage path in PRN}, we obtain that $\sigma^R$ is the most upstream simple node in $\RR$ which allows an appendage path from 
$\til{\AA}^{\RR}$. Thus, we conclude that
\begin{equation} \notag
\til{\AA}^{\RR} \to \tLL (\rho_i^P, \rho_{i+1}^R)
\end{equation}
The remaining items follow directly from the first two items, so we omit the proof.
\end{proof}

\begin{lemma} \label{arrow between structural pattern network and appendage component 2}
Let $\GG$ be an input-output GRN and $\RR$ be the associated input-output PRN.
Consider $\VV \in \PP_{\sS} (\GG)$ and $\til{\AA} \in \PP_{\AA} (\GG)$. 
Suppose $\VV \to \til{\AA} \in \PP (\GG)$, then the corresponding arrows in $\PP (\RR)$ are
\begin{enumerate}[{\rm (a)}]
\item If $\VV = \rho_i$ and $\til{\AA} \in \{ \til{\AA}'_i \}^{r}_{i=1}$, then
\begin{equation}
\label{arrows from structural to appendage in R part1}
\rho^P_i \to \til{\AA}^{R} \to \til{\AA}^{P}
\end{equation}

\item If $\VV = \tLL_i$ and $\til{\AA} \in \{ \til{\AA}_j \}^{s}_{j=1}$, then
\begin{equation}
\label{arrows from structural to appendage in R part2}
\tLL (\rho_i^P, \rho_{i+1}^R) \to \til{\AA}^{\RR}
\end{equation}

\item If $\VV = \rho_i$ and $\til{\AA} \in \{ \til{\AA}_j \}^{s}_{j=1}$, then
\begin{equation}
\label{arrows from structural to appendage in R part3}
\rho^P_i \to \til{\AA}^{\RR}
\end{equation}

\item If $\VV = \tLL_i$ and $\til{\AA} \in \{ \til{\AA}'_i \}^{r}_{i=1}$, then
\begin{equation}
\label{arrows from structural to appendage in R part4}
\tLL (\rho_i^P, \rho_{i+1}^R) \to \til{\AA}^{R} \to \til{\AA}^{P}
\end{equation}
\end{enumerate}
\end{lemma}

\begin{proof}
(a) Similarly as in Lemma \ref{arrow between structural pattern network and appendage component 1}, we assume $\til{\AA} = \{\tau_1 \}$ with $\tau_1$ is a single appendage node in $\GG$, and get the corresponding appendage components in $\PP_{\AA} (\RR)$ as
\[
\til{\AA}^{R} = \{\tau^R_1 \}, \; \til{\AA}^{P}_1 = \{\tau^P_1\}
\] 
and the nodes in $\PP_\sS (\RR)$ as
\[
\rho^R_i, \; \til{\LL} (\rho^R_i, \rho^P_i), \; \rho^P_i
\]
From the arrow $\rho_i \to \tau_1$, $\rho_i$ is the most downstream simple node in $\GG$ that allows an appendage path to $\tau_1$.
Using appendage paths correspondence in Corollary \ref{appendage path in PRN}, we obtain that $\rho_i^P$ is the most downstream simple node in $\RR$ that allows an appendage path to $\til{\AA}^{R}$.
Together with the gene coupling in the PRN, we obtain arrows 
\[
\rho^P_i \to \til{\AA}^{R} \to \til{\AA}^{P}
\]

\smallskip
\noindent
(b) Follows from Lemma \ref{arrow between structural pattern network and appendage component 1}, we get the corresponding appendage component
$\til{\AA}^{\RR}$ in $\PP_{\AA} (\RR)$ 
and the backbone node $\tLL (\rho_i^P, \rho_{i+1}^R)$ in $\PP_\sS (\RR)$.
Moreover, from the arrow $\tLL_i \to \til{\AA}$ in $\PP (\GG)$, there exists two gene nodes $\sigma \in \tLL_i$ and $\tau \in \til{\AA}$, such that 
\begin{equation} \notag
\sigma \to \tau
\end{equation}
where $\sigma$ is the most downstream simple node in $\GG$ that allows an appendage path to nodes in appendage component $\til{\AA}$. 
Using Lemma \ref{subnetwork correspondence} and the gene coupling in PRN, we have 
\begin{equation} \notag
\{ \sigma^R, \sigma^P \} \in \tLL (\rho_i^P, \rho_{i+1}^R), \; \{ \tau^R, \tau^P \} \in \til{\AA}^{\RR}, \; \text{and } \sigma^P \to \tau^R
\end{equation}
Again from Corollary \ref{appendage path in PRN}, we obtain that $\sigma^P$ is the most downstream simple node in $\RR$ that allows an appendage path to 
$\til{\AA}^{\RR}$. Thus, we conclude that
\begin{equation} \notag
\tLL (\rho_i^P, \rho_{i+1}^R) \to \til{\AA}^{\RR}
\end{equation}
Items (c) and (d) follow directly from the above.
\end{proof}

\subsection{Homeostasis Inducing in GRN and PRN}
\label{SS: homeo inducing in the GRN}

Before stating the second main result of the paper we need more terminology.

\begin{definition} \normalfont
Let $\GG$ be a GRN and $\RR$ the associated PRN.
Let $\KK$ be a homeostasis subnetwork or a super-simple node of $\GG$.
\begin{enumerate}[(a)]
\item We say that $\KK$ is \emph{homeostasis inducing} if $h_{\KK^{\RR}} \equiv \det(B^{\RR}) = 0$ at $\II_0$, where $\KK^{\RR}$ and $B^{\RR}$ are the corresponding homeostasis subnetwork and homeostasis block in the PRN.
If $\KK = \{\tau\}$, for a single appendage node $\tau$, then $\KK^{\RR}=\{\tau^R\}$ and $B^{\RR}=[f_{\tau^R,\tau^R}]$.
\item If $\KK = \{\rho\}$ is a super-simple node of $\GG$, let $\rho \Rightarrow \nu \in \GG$ denote that nodes $\nu^R, \nu^P \in \RR$ are generically homeostatic whenever $\rho$ is homeostasis inducing. 
Given a subset of nodes $\NN\subset\GG$, $\rho \Rightarrow \NN$ if for every node $\nu \in \NN$, $\rho \Rightarrow \nu$.
\item If $\KK$ is a homeostasis subnetwork of $\GG$, let $\widetilde{\KK}$ be the corresponding node in $\PP (\GG)$. 
Then, let $\widetilde{\KK} \Rightarrow \nu \in \GG$ denote that nodes $\nu^R, \nu^P \in \RR$ are generically homeostatic whenever $\KK$ is homeostasis inducing.
Given a subset of nodes $\NN\subset\GG$, $\widetilde{\KK} \Rightarrow \NN$ if for every node $\nu \in \NN$, $\widetilde{\KK} \Rightarrow \nu$.
\END
\end{enumerate}

\end{definition}

Now we are ready to use the homeostasis pattern network $\PP(\GG)$ to characterize the homeostasis patterns in both the GRN $\GG$ and the associated PRN $\RR$. 

\begin{theorem} 
\label{thm: patterns of homeostasis in RPN}
Let $\GG$ be a GRN and $\PP(\GG)$ its homeostasis pattern network.
Suppose that $\til{\AA} \in \PP_{\AA} (\GG)$ is an appendage component, and $\VV_s \in \PP_{\sS}(\GG)$ is a backbone node.
Then
\begin{enumerate}[{\rm (a)}]
\item {\em (Structural Homeostasis $\Rightarrow$ Structural Subnetwork)}
$\VV_s$ induces precisely every node of $\PP_{\sS}(\GG)$ strictly downstream from $\VV_s$. 
Moreover, if $\VV_s = \{ \rho \}$ is a super-simple node, then $\rho^P \in \RR$ is also homeostatic.

\item {\em (Structural Homeostasis $\Rightarrow$ Appendage Subnetwork)}
Let $\VV \to \til{\AA} \in \PP(\GG)$ with $\VV \in \PP_\sS(\GG)$. 
$\VV_s \Rightarrow \til{\AA}$ if and only if $\VV$ is strictly downstream from $\VV_s$.
If $\VV_s = \{ \rho \}$ is a super-simple node, then 
$\rho \Rightarrow \til{\AA}$ if and only if $\VV$ is downstream from $\VV_s$.

\item {\em (Appendage Homeostasis $\Rightarrow$ Structural Subnetworks)}
Let $\til{\AA} \to \VV \in \PP (\GG)$ with $\VV \in \PP_\sS (\GG)$. 
$\til{\AA} \Rightarrow \VV_s$ if and only if $\VV$ is strictly upstream from $\VV_s$.
If $\VV_s = \{ \rho \}$ is a super-simple node, then 
$\til{\AA} \Rightarrow \rho$ if and only if $\VV$ is upstream from $\rho$.

\item {\em (Appendage Homeostasis $\Rightarrow$ Appendage Subnetworks)}
Let $\til{\AA}_1, \til{\AA}_2 \in \PP_\AA (\GG)$ be distinct appendage components. 
Let $\til{\AA}_1 \to \VV_1, \VV_2 \to \til{\AA}_2 \in \PP (\GG)$ with $\VV_1, \VV_2 \in \PP_\sS (\GG)$.
$\til{\AA}_1 \Rightarrow \til{\AA}_2$ if and only if 
$\til{\AA}_1$ is upstream from $\til{\AA}_2$ and every path from $\til{\AA}_1$ to $\til{\AA}_2$ in $\PP (\GG)$ contains a super-simple node which is downstream from $\VV_1$ and upstream from $\VV_2$.
Moreover, if $\til{\AA}_1 = \{ \tau \}$ is a single appendage node, then P-null-degradation induces $\tau^R \in \RR$, but R-null-degradation does not induce $\tau^P \in \RR$.
\end{enumerate}
\end{theorem}

\begin{proof}
Here we apply Lemmas \ref{node in homeostasis pattern network} - \ref{arrow between structural pattern network and appendage component 2} to obtain the nodes and arrows in $\PP (\RR)$, including the arrows in $\PP_\sS (\RR)$ or $\PP_{\AA} (\RR)$, and the arrows between $\PP_\sS (\RR)$ and $\PP_{\AA} (\RR)$. Then we use the results from Theorem \ref{thm:pattern of homeo}.

\smallskip
\noindent
(a) Under assumptions on the GRN $\GG$, Lemmas \ref{node in homeostasis pattern network} and \ref{arrow of structural pattern network or appendage component} show that the nodes of the structural pattern network $\PP_\sS (\RR)$ satisfy
\begin{equation} \notag
\rho^R_1 \to \tLL (\rho_1^R,\rho_1^P) \to \rho^P_1 \to \tLL (\rho_1^P,\rho_2^R) \to \rho^R_2 \to
 \cdots \to \tLL (\rho_{q+1}^R,\rho_{q+1}^P) \to \rho^P_{q+1}
\end{equation}
Suppose $\VV_s$ is homeostasis inducing. Using Theorem \ref{thm:pattern of homeo}$(a)$, we derive that $\VV_s$ induces precisely every node of the structural pattern network strictly downstream from $\VV_s$. 
Moreover, if $\VV_s = \{ \rho \}$ is a super-simple node it follows that homeostasis is induced by $\RR$-Haldane 
$\big[ f_{\rho^P,  \rho^R} \big]$, thus $\rho^P \in \RR$ is also homeostatic.

\smallskip
\noindent
(b) Given $\VV \to \til{\AA} \in \PP (\GG)$, Lemma \ref{arrow between structural pattern network and appendage component 2} shows the corresponding arrows from $\PP_\sS (\RR)$ to $\PP_{\AA} (\RR)$. 
Thus, if $\VV = \{ \rho \}$ is a super-simple node, then
\begin{equation} \notag
\rho^P_i \to \til{\AA}^{R} \to \til{\AA}^{P}
\ \text{ or } \
\rho^P_i \to \til{\AA}^{\RR}
\end{equation}
If  $\VV = \tLL_i$ is a backbone node, then
\begin{equation} \notag
\tLL (\rho_i^P, \rho_{i+1}^R) \to \til{\AA}^{R} \to \til{\AA}^{P}
\ \text{ or } \
\tLL (\rho_i^P, \rho_{i+1}^R) \to \til{\AA}^{\RR}
\end{equation}
Using Lemma \ref{arrow of structural pattern network or appendage component} and Theorem \ref{thm:pattern of homeo}$(b)$, we get $\VV_s \Rightarrow \til{\AA}$ if and only if $\VV$ is strictly downstream from $\VV_s$. Moreover, if $\VV_s = \{ \rho \}$ is a super-simple node, $\VV$ needs only to be downstream from $\rho$.

\smallskip
\noindent
(c) From Lemma \ref{arrow between structural pattern network and appendage component 1}, given $\til{\AA} \to \VV \in \PP (\GG)$, we obtain the corresponding arrows 
from $\PP_{\AA} (\RR)$ to $\PP_\sS (\RR)$. 
Hence, if $\VV = \{ \rho \}$ is a super-simple node, then
\begin{equation} \notag
\til{\AA}^{R} \to \til{\AA}^{P} \to  \rho^R_i 
\ \text{ or } \
\til{\AA}^{\RR} \to  \rho^R_i
\end{equation}
If  $\VV = \tLL_i$ is a backbone node, then
\begin{equation} \notag
\til{\AA}^{R} \to \til{\AA}^{P} \to \tLL (\rho_i^P, \rho_{i+1}^R)
\ \text{ or } \
\til{\AA}^{\RR} \to \tLL (\rho_i^P, \rho_{i+1}^R)
\end{equation}
Together with Lemma \ref{arrow of structural pattern network or appendage component} and Theorem \ref{thm:pattern of homeo}$(c)$, we derive $\til{\AA} \Rightarrow \VV_s$ if and only if $\VV$ is strictly upstream from $\VV_s$. Moreover, if $\VV_s = \{ \rho \}$ is a super-simple node, we need only that $\VV$ is upstream from $\rho$.

\smallskip
\noindent
(d) First, suppose $\til{\AA}_1, \til{\AA}_2 \in \PP_\AA (\GG)$ are distinct appendage components.
From Lemmas \ref{arrow of structural pattern network or appendage component} - \ref{arrow between structural pattern network and appendage component 2}, $\til{\AA}_1$ being upstream from $\til{\AA}_2$ in $\PP (\GG)$ is equivalent to the corresponding node of $\til{\AA}_1$ being upstream from the corresponding node of $\til{\AA}_2$ in $\PP (\RR)$.
In addition, Lemma \ref{simple path correspondence} shows that every path from $\til{\AA}_1$ to $\til{\AA}_2$ in $\PP (\GG)$ containing a super-simple node is equivalent to every path from the corresponding node of $\til{\AA}_1$ to the corresponding node of $\til{\AA}_2$ in $\PP (\RR)$ containing a super-simple node in PRN $\RR$.
Using Theorem \ref{thm:pattern of homeo}$(d)$, we conclude this part.
Second, suppose $\til{\AA}_1 = \{\tau\}$ is a single appendage node and homeostasis inducing, that is, the homeostasis is induced by either R-null-degradation $\big[ f_{\tau^R, \tau^R} \big]$ or P-null-degradation $\big[ f_{\tau^P, \tau^P} \big]$.
Since $\tau^R \to \tau^P \in \PP (\RR)$ is an appendage path, from Theorem \ref{thm:pattern of homeo}$(d)$ we get that R-null-degradation doesn't induce $\tau^P \in \RR$.
On the other side, since $\tau$ is a single appendage node, $\tau^P \to \tau^R \notin \PP (\RR)$. Theorem \ref{thm: homeo subnetwork in PRN} and Lemmas \ref{arrow between structural pattern network and appendage component 1} - \ref{arrow between structural pattern network and appendage component 2} imply that there exists at least one path from $\tau^P$ to $\tau^R$ and every such path must contain a super-simple node in $\PP (\RR)$, therefore P-null-degradation induces $\tau^R \in \RR$.
\end{proof}

The following Lemma shows the connection between the occurrence of homeostasis on the mRNA node and on the protein node associated to the same gene node in GRN.

\begin{lemma} \label{lem: connection on pattern of homeo in the GRN}
Let $\GG$ be an input-output GRN and $\RR$ be the associated input-output PRN.
Suppose that infinitesimal homeostasis occurs in the PRN at $\II_0$, induced by a homeostasis subnetwork $\KK$ of $\RR$.
Then
\begin{enumerate}[{\rm (a)}]
\item If $\nu$ is neither a super-simple nor single appendage node in $\GG$, with $\nu^R,\nu^P\notin\KK$, then $\nu^R$ is homeostatic if and only if $\nu^P$ is homeostatic in $\RR$.

\item If $\rho$ is a super-simple node in $\GG$, with $\rho^R,\rho^P\notin\KK$, then generically $\rho^R$ is homeostatic if and only if $\rho^P$ is homeostatic in $\RR$.

\item If $\tau$ is a single appendage node in $\GG$, with $\tau^R,\tau^P\notin\KK$, then generically $\tau^R$ is homeostatic if and only if $\tau^P$ is homeostatic in $\RR$.
\end{enumerate}
\end{lemma}

\begin{proof}
(a) Using Lemma \ref{node in homeostasis pattern network} and the fact that the GRN-node $\nu$ is neither a super-simple nor single appendage node in $\GG$, we have that $\nu^R, \nu^P \in \RR$ belong to the same homeostasis subnetwork of $\RR$, distinct from $\KK$.
Thus, we conclude the result from Theorem \ref{thm: patterns of homeostasis in RPN}.

\smallskip
\noindent
(b) By Lemma \ref{node in homeostasis pattern network} and Lemma \ref{arrow of structural pattern network or appendage component} and the fact that $\rho$ is a super-simple node in $\GG$, we have  
\begin{equation}
\rho^R \to \tLL (\rho^R,\rho^P) \to \rho^P
\end{equation}
where $\rho^R, \tLL (\rho^R,\rho^P), \rho^P \in \PP_\sS (\RR)$ and $\tLL (\rho^R, \rho^P) = \emptyset$.
Assume that $\rho^R$ is homeostatic in the PRN. 
Then, we have the steady-state equation $f_{\rho^P}(\rho^R,\rho^P)=0$ and implicit differentiation gives
\begin{equation} \label{eq:derivative_of_function}
\frac{d}{d \II} f_{\rho^P} \big|_{\II = \II_0} =   f_{\rho^P, \rho^R} (\rho^R)' +  f_{\rho^P, \rho^P} (\rho^P)' = 0
\end{equation}
Generically, we can assume $f_{\rho^P, \rho^P} \neq 0$ (i.e. no null degradation, since $\rho$ is not appendage), and get
\begin{equation} \notag
(\rho^P)' = 0
\end{equation}
This shows that $\rho^P$ is homeostatic in the PRN.
On the other hand, assume that $\rho^P$ is homeostatic in the PRN. 
Again, \eqref{eq:derivative_of_function} holds.
Generically, we can assume $f_{\rho^P, \rho^R} \neq 0$ (i.e. no Haldane homeostasis, since $\rho^R,\rho^P\notin\KK$), and obtain 
\begin{equation} \notag
(\rho^R)' = 0
\end{equation}
This shows that $\rho^R$ is homeostatic in the PRN.

\smallskip
\noindent
We skip the proof of item (c), since it is analogous to the proof of item (b). 
\end{proof}

\begin{remarks} \normalfont 
\label{rmk:homeo_fail_GRN}
In Lemma \ref{lem: connection on pattern of homeo in the GRN}, items (b) and (c) characterize homeostatic nodes that are not contained in the homeostasis subnetwork $\KK$ inducing homeostasis.
Now, it is easy to see what happens when they do belong to $\KK$.
\begin{enumerate}[(a)]
\item Suppose that $\rho$ is a super-simple node in $\GG$.
On one hand, if $\rho^R\in\KK$ then $\RR$-Haldane homeostasis occurs at $\II_0$ (i.e. $f_{\rho^P,\rho^R} = 0$).
In this case, $\rho^R$ is not homeostatic but $\rho^P$ is homeostatic.
On the other hand, if $\rho^P\in\KK$
then $\rho^R$ and $\rho^P$ fail to be simultaneously homeostatic.
In both cases, $\rho$ is not GRN-homeostatic.

\item Suppose that $\tau$ is a single appendage node in $\GG$.
On one hand, if $\tau^P\in\KK$ then $P$-null-degradation occurs at $\II_0$ (i.e. $f_{\tau^P,\tau^P} = 0$).
In this case, $\tau^P$ is not homeostatic but $\tau^R$ is homeostatic.
On the other hand, if $\tau^R\in\KK$ ($R$-null-degradation) then $\tau^R$ and $\tau^P$ fail to be simultaneously homeostatic.
In both cases, $\tau$ is not GRN-homeostatic.
\END
\end{enumerate}
\end{remarks}

\end{document}